\documentclass[11pt]{article}

\usepackage[utf8]{inputenc}
\usepackage{fullpage,amsfonts,bm}
\usepackage{amsmath,amssymb,amsthm,mathtools}
\usepackage{mleftright}\mleftright
\usepackage[pagebackref,final]{hyperref}
\usepackage{enumitem}
\usepackage{algpseudocode}
\usepackage{float}

\usepackage{verbatim}
\usepackage{ifdraft}
\ifdraft{\newcommand{\authnote}[3]{{\color{#3} {\bf  #1:} #2}}}{\newcommand{\authnote}[3]{}}
\newcommand{\anote}[1]{\authnote{\text{ András}}{#1}{green}}

\newcommand{\jnote}[1]{\authnote{ J}{#1}{red}}

\newcommand{\shownote}[1]{{\color{blue} {\bf  NOTE:} #1}}

\usepackage{tikz}
\usetikzlibrary{calc,arrows,quotes,angles}

\usepackage{subcaption,multirow}

\usepackage{footnote}
\newcommand{\footref}[1]{\textsuperscript{\textup{#1}}}

\usepackage{thmtools} 
\usepackage{thm-restate}

\floatstyle{boxed}

\newfloat{algorithm}{t}{lop}
\floatname{algorithm}{Algorithm} 

\newfloat{metaalgorithm}{t}{lop}
\floatname{metaalgorithm}{Meta-Algorithm} 

\newtheorem{theorem}{Theorem} 

\newtheorem{lemma}[theorem]{Lemma}

\newtheorem{corollary}[theorem]{Corollary}
\newtheorem{definition}[theorem]{Definition}

\mathchardef\mhyphen="2D

\newcommand{\ket}[1]{|#1\rangle}
\newcommand{\bra}[1]{\langle#1|}
\newcommand{\braket}[2]{\langle#1|#2\rangle}

\newcommand{\ketbra}[2]{|#1\rangle\! \langle #2|}

\newcommand{\Tr}{\mbox{\rm Tr}}
\newcommand{\tr}[1]{\Tr\left(#1\right)}
\newcommand{\opt}{\mbox{\rm OPT}}
\newcommand{\polylog}{\mbox{\rm polylog}}
\newcommand{\poly}{\mbox{\rm poly}}

\newcommand{\F}{\mathcal{F}}
\newcommand{\R}{\mathcal{R}}

\newcommand{\eps}{\varepsilon}
\newcommand{\nrm}[1]{\left\lVert#1\right\rVert}
\def\01{\{0,1\}}

\newcommand{\diag}{\mbox{\rm diag}}

\newcommand{\bigO}[1]{\mathcal{O}\left( #1 \right)}

\newcommand{\bOt}[1]{\widetilde{\mathcal O}\left(#1\right)}
\newcommand{\BOT}[2]{\widetilde{\mathcal O}_{#1}\left(#2\right)}

\newcommand{\srho}{{ \varrho}}
\newcommand{\asrho}{{ \tilde{ \varrho}}}
\newcommand{\ssigma}{{ \varsigma}}
\newcommand{\assigma}{{ \tilde{ \varsigma}}}

\begin{document}

\title{Improvements in Quantum SDP-Solving with Applications }
\author{Joran van Apeldoorn\thanks{QuSoft, CWI, the Netherlands. Supported by the Netherlands Organization for Scientific Research, grant number 617.001.351. {\tt apeldoor@cwi.nl}}
  \and
  Andr\'as Gily\'en\thanks{QuSoft, CWI, the Netherlands.
    Supported by ERC Consolidator Grant 615307-QPROGRESS. {\tt gilyen@cwi.nl}}
}
\date{}
\maketitle

\begin{abstract}
\noindent Following the first paper on quantum algorithms for SDP-solving by Brand\~ao and Svore~\cite{brandao:quantumsdp} in 2016, rapid developments has been made on quantum optimization algorithms. In 2017 van Apeldoorn et al.~\cite{AGGW:SDP} improved the quantum algorithm introduced by~\cite{brandao:quantumsdp} and gave stronger lower bounds as well. Recently Brand\~ao et al.~\cite{brandao:expSDP2} improved the quantum SDP-solver in the so-called quantum state input model, where the input matrices of the SDP are given as purified mixed states. They also gave the first non-trivial application of quantum SDP-solving by obtaining a more efficient algorithm for the problem of shadow tomography~\cite{aaronson:shadow}.

In this paper we improve on all previous quantum SDP-solvers. Mainly we construct better Gibbs-samplers for both input models, which directly gives better bounds for SDP-solving. We also combine the Fast Quantum OR lemma of Brand\~ao et al.~\cite{brandao:expSDP2} and the Gentle Quantum Search Lemma of Aaronson~\cite{aaronson:shadow}, and use the techniques from van Apeldoorn et al.~\cite{AGGW:SDP} to give an improved general quantum SDP-solving framework. For an SDP with $m$ constraints involving $n\times n$ matrices, our improvements yield an $\bOt{  \left( \sqrt{m}  +  \sqrt{n}\gamma  \right)s \gamma^4}$ upper bound on SDP-solving in the sparse matrix input model and an $\bOt{ \left(\sqrt{m}+B^{2.5}\gamma^{3.5} \right)B\gamma^4 }$ upper bound in the quantum state input model. Here $\gamma = Rr/\eps$ is the additive error $\eps$ scaled down with bounds $R$ and $r$ on the size of optimal solutions, $s$ is the row-sparsity of the input matrices in the sparse matrix input model and $B$ is a normalization factor for the input states in the quantum state model. We also introduce the quantum operator input model, which generalizes both other input models. In this more general model we give an $\bOt{  \left( \sqrt{m}  +  \sqrt{n}\gamma  \right)\alpha \gamma^4}$-query algorithm, where $\alpha$ is a normalization factor of the input operators. 

We then apply these results to the problem of shadow tomography to simultaneously improve the best known upper bounds on sample complexity~\cite{aaronson:shadow} and complexity~\cite{brandao:expSDP2}. Furthermore, we apply our quantum SDP-solvers to the problems of quantum state discrimination and E-optimal design. In both cases we beat the classical lower bound in terms of some parameters, at the expense of heavy dependence on some other parameters.

Finally we prove two lowers bounds for solving SDPs using quantum algorithms: (1) $\tilde{\Omega}(\sqrt{m}B/\eps)$ in the quantum state input model, and (2) $\tilde{\Omega}(\sqrt{m}\alpha/\eps)$ in the quantum operator input model. These lower bounds show that the $\sqrt{m}$ factor and the polynomial dependence on the parameters $B,\alpha$, and $1/\eps$ are necessary. 
\end{abstract}

\newpage

\section{Introduction} 
\subsection{Semidefinite programs} \label{subsec:sdp} 
In this paper we consider \emph{Semidefinite programs} (SDPs). SDPs have many applications in optimization, notable examples include approximation of NP-hard problems like MAXCUT~\cite{GoemansWilliamson95} and polynomial optimization through the Sum-Of-Squares hierarchy~\cite{lasserre:sos,parrilo:thesis}. SDPs have also found applications in quantum information theory. Examples include POVM measurement design~\cite{eldar:sdp} and finding the winning probability of non-local games~\cite{cleve2004cal}.

We consider the basic (primal) form of an SDP as follows:
\begin{align} \label{eq:SDP}
\opt = \max \quad &\Tr(CX) \\ 
\text{s.t.}\ \ \ &\Tr(A_j X) \leq b_j \quad \text{ for all } j \in [m], \notag \\
&X \succeq 0, \notag
\end{align}
where $[m]:=\{1,\ldots,m\}$.
The input to the problem consists of $n\times n$ Hermitian constraint matrices $A_1,\ldots,A_m$, an objective matrix $C$ and reals $b_1,\ldots,b_m$.  For normalization purposes we assume $\nrm{C},\nrm{A_j} \leq 1$.
The number of constraints is~$m$ (we do not count the standard $X\succeq 0$ constraint for this). The variable~$X$ of this SDP is an $n\times n$ positive semidefinite (psd) matrix. 
We assume that $A_1 = I$ and $b_1 = R$, giving a known bound on the trace of a solution: $\tr{X} \leq R$. 
A primal SDP also has a \emph{dual}. For a primal SDP of the above form~\eqref{eq:SDP} the dual SDP is
\begin{align} \label{eq:SDP2}
\opt =   \min \quad & b^T y \\ \notag
  \text{s.t.}\ \ \ &\sum_{j=1}^m y_j A_j - C \succeq 0,\\ \notag
             &y \geq 0.
\end{align}
We assume that the dual optimum is attained and that an explicit $r\geq 1$ is known such that at least one optimal dual solution $y$ exists $\nrm{y}_1\leq r$. 
These assumptions imply that strong duality holds, justifying the use of $\opt$ for both optimal values. Linear programs (LPs) correspond to the case where all constraint  matrices are diagonal.

In this paper we build on the observation that a normalized psd matrix can  be naturally represented as a quantum state. Since operations on quantum states can sometimes be cheaper to perform on a quantum computer than operations on classical descriptions of matrices, this can give rise to faster algorithms for solving SDPs on a quantum computer~\cite{brandao:quantumsdp}.

We say an algorithm is an $\eps$-approximate \emph{quantum SDP-solver} if for all input numbers $g\in\mathbb{R}$ and $\zeta\in(0,1)$, with success probability $1-\zeta$, all of the following hold:
\begin{itemize}
\item The algorithm determines whether $\opt\leq g -\eps$ or $\opt\geq g +\eps$. If $\opt \in [g-\eps,g+\eps]$ then it may output either.
\item The algorithm finds a $y\in \mathbb{R}^{m+1}$ that is an \emph{$\eps$-feasible solution} to the dual problem with objective value at most $g +\eps$, i.e.,
\[
\sum_{j=1}^m y_jA_j -C \succeq -\eps I,
\]
and $\langle y,b\rangle \leq g+\eps$, or it concludes that no such $y$ exists even if we would set $\eps = 0$. 
\item The algorithm finds a vector $y'\in \mathbb{R}^{m+1}$ and a real number $z$ such that for 
\begin{equation}\label{eq:hardtowrite}
\rho:= \frac{e^{-\sum_{j=1}^m y'_jA_j + y'_0 C} }{ \tr{e^{-\sum_{j=1}^m y'_jA_j + y'_0 C}}}
\end{equation}
 we have that $z\rho$ is an \emph{$\eps$-feasible primal solution} with objective value at least $g- \eps$, i.e.,
\[
  \forall j\in [m]\colon\tr{z\rho A_j } \leq b_j +\eps,
\]
and $\tr{z \rho C} \geq g - \eps$, or concludes that no such $z$ and $y'$ exist even if we would set $\eps = 0$. 
\end{itemize}
 Notice that we can easily find an approximation of $\opt$ using binary search on $g$ if we have an $\eps$-approximate SDP-solver. An algorithm that only satisfies the last of the three points will be called an \emph{$\eps$-approximate SDP primal oracle}. Due to the form of the objective value constraint in this last point, and to simplify statements like~\eqref{eq:hardtowrite}, we write $A_0 := -C$ and $b_0:= -g$.

In Sections~\ref{subsec:Hamsim} and~\ref{subsec:stategibbs} we will work with \emph{subnormalized density operators}:
\begin{definition}[Subnormalized density operators \& Purification]
A \emph{subnormalized density operator} $\srho$ is a psd matrix of trace at most $1$. 

A \emph{purification} of a subnormalized density operator $\srho$ is a pure state consisting of $3$ registers such that tracing out the third register\footnote{\label{notePurification}For simplicity we assume that for a $d$-dimensional density operator a purification has at most polylog$(d)$ qubits.} and projecting on the subspace where the second register is $\ket{0}$ yields $\srho$. 

We write ``$\srho$'' and ``$\ssigma$'' for subnormalized density operators to distinguish them from normalized density operators, for which we write ``$\rho$'' and ``$\sigma$''. 
\end{definition}

\paragraph{Notation.} We use the following definition for $\tilde{\mathcal{O}}$:
\[
\BOT{d,e}{f(a,b,c)} := \bigO{f(a,b,c) \cdot \polylog(f(a,b,c),d,e)}.
\]
We define $\tilde{\Omega}$ in a similar way and $\tilde{\Theta}$ as the intersection of the two.
We write $\delta_{ij}$ for the Kronecker delta function and $e_j$ for the $j$th basis vector in the standard basis when the dimension of the space is clear from context. For a Hermitian matrix $H$ we write $\mathrm{Spec}(H)$ for its spectrum (set of eigenvalues). For a function $f:\mathbb{R}\rightarrow\mathbb{R}$ we write $f(H)$ for the matrix we get by applying $f$ to the eigenvalues of $H$, i.e., 
\[
f(H) = U \begin{bmatrix} f(\lambda_1) & & \\ & \ddots & \\ & & f(\lambda_n)\end{bmatrix} U^{-1} \text{ where } H =  U \begin{bmatrix} \lambda_1 & & \\ & \ddots & \\ & & \lambda_n\end{bmatrix} U^{-1}.
\]

\subsection{Input models \& Subroutines} 
We will consider three input models: the \emph{sparse matrix model}, the \emph{quantum state model}, and the \emph{quantum operator model}. In all models we assume quantum oracle access to the numbers $b_j$ via the input oracle $O_{b}$ satisfying\footnote{\label{noteBits}For simplicity we assume the bitstring representation has at most $\bigO{\log(nm R r/\eps)}$ bits.} for all $j\in[m]\colon$
$$O_{b} \ket{j}\ket{0} = \ket{j}\ket{b_j}.$$
For all input oracles we assume we can apply both the oracle and its inverse\footnote{When we talk about samples, e.g. in Section~\ref{sec:shadow}, then we do not assume we can apply the inverse operation.} in a controlled fashion. 

\paragraph{Sparse matrix model.}
In the \emph{sparse matrix model} the input matrices are assumed to be $s$-row sparse for a known bound $s\in [n]$, meaning that there are at most $s$ non-zero elements per row. Access to the $A_j$ matrices is provided by two oracles, similar to previous work on Hamiltonian simulation in~\cite{BerryChilds:hamsimFOCS}. The first of the two oracles is a unitary $O_{\text{sparse}}$, which serves the purpose of sparse access. This oracle calculates the $\mathbf{index}: [m] \times [n] \times [s] \to [n]$ function, which for input $(j,k,\ell)$ gives the column index of the $\ell$th non-zero element in the $k$th row of $A_j$. We assume this oracle computes the index ``in place":
\begin{equation}
  O_{\text{sparse}} \ket{j,k,\ell} = \ket{j,k,\mathbf{index}(j,k,\ell)}.
  \label{eq:oracleind}
\end{equation}
(In the degenerate case where the  $k$th row has fewer than $\ell$ non-zero entries, $\mathbf{index}(j,k,\ell)$ is defined to be $\ell$ together with some special symbol indicating this case.) 

We also need another oracle $O_A$, returning a bitstring\footref{\ref{noteBits}} representation of $(A_j)_{ki}$ for every $j \in [m]$ and $k,i \in [n]$:
\begin{equation}
  O_A\ket{j,k,i,z} = \ket{j,k,i,z \oplus (A_j)_{ki} }.
  \label{eq:oraclemat}
\end{equation}

This model corresponds directly to a classical way of accessing sparse matrices.
\paragraph{Quantum state model.}
In contrast to the sparse matrix model, the \emph{quantum state model} is inherently quantum and has no classical counterpart for SDPs.\footnote{However, there is a natural classical analogue in the case of LPs, when the constraints $a_j$ are given by random variables that outputs $k$ with probability proportional to $a_{jk}$. In this case a classical algorithm with complexity $\BOT{n}{m\poly(B,\gamma)}$ is possible. In a similar manner, the classical input model can be sped up for classical LP solvers as well using techniques similar to those presented in this paper, which would lead to a $\bOt{(n+m)\poly(\gamma)}$ algorithm.}
In this model we assume that each $A_j$ has a fixed decomposition of the form
\[
A_j = \mu^+_j \srho^+_j - \mu^-_j \srho^-_j + \mu^I_j I 
\] 
for (subnormalized) density operators $\srho^{\pm}_j$, non-negative reals $\mu^{\pm}_j$ and real number $\mu^I_j\in\mathbb{R}$.
We assume access to an oracle $O_{\mu}$ that takes as input an index $j$ and outputs binary representations\footref{\ref{noteBits}} of $\mu^+_j,\mu^-_j$ and $\mu^I_j$.

Furthermore we assume access to a state-preparing oracle $O_{\ket{\cdot }}$ that prepares purifications\footref{\ref{notePurification} }$\ket{\psi^{\pm}_j}$ of $\srho^{\pm}_j$:
\[
 O_{\ket{\cdot }} \ket{j}\ket{\pm} \ket{0}  = \ket{j}\ket{\pm}\ket{\psi^{\pm}_j}.
\]
Finally we assume that a bound $B\in\mathbb{R}_+$ is known such that
\[
\forall j:  \mu^+_j+\mu^-_j + |\mu^I_j|\leq B.
\]
Note that a tight upper bound $B$ can easily be found using $\bigO{\sqrt{m}}$ quantum queries to $O_{\mu}$ by means of maximum finding~\cite{durr&hoyer:minimum}. 
\paragraph{Quantum operator model.}
We propose a new input model that we call the \emph{quantum operator model}. In this model the input matrices are given by a unitary that implements a block-encoding:
\begin{definition}[Block encoding]\label{def:standardForm}
	Suppose that $A$ is a $w$-qubit operator, $\alpha,\eps\in\mathbb{R}_+$ and $k\in \mathbb{N}$, then we say that the $(a+w)$-qubit unitary $U$ is an $(\alpha,a,\eps)$-block-encoding of $A$, if 
	$$ \nrm{A - \alpha(\bra{0}^{\otimes a}\otimes I)U(\ket{0}^{\otimes a}\otimes I)}\leq \eps.$$
\end{definition}
Roughly speaking this means that $A$ is represented by a unitary
$$ U\approx\left(\begin{array}{cc} A/\alpha & . \\ . & .\end{array}\right). $$
In the quantum operator model we assume access to an oracle $O_{U}$ that acts as follows:
\[
O_{U} \ket{j}\ket{\psi} = \ket{j}(U_j\ket{\psi}).
\]
Where $U_j$ is an $(\alpha,a,0)$-block-encoding\footnote{If $n$ is not a power of $2$, then we simply define $A_j$ to be zero on the additional $2^w-n$ dimensions.} of $A_j$, for some fixed\footnote{Having a single normalization parameter $\alpha$ is not a serious restriction as it is easy to make a block-encoding more subnormalized so that every $A_j$ gets the same normalization, cf.~Lemma~\ref{lem:linCombBlocks}.} $\alpha\in \mathbb{R}$ and $a = \bigO{\log(nm R r/\eps)}$.

In Section~\ref{subsec:Hamsim} we will show that the sparse input model can be reduced to the quantum operator model with $\alpha = s$ and that the quantum state model can be reduced to it with $\alpha = B$. We will also argue that if we can perform a measurement corresponding to $A_j\succeq 0$ using $a$ ancilla qubits, i.e., accept a state $\rho$ with probability $\tr{A_j\rho}$, then we can implement a $(1,a+1,0)$-block-encoding of $A_j$. 

\paragraph{Computational cost.} 
We will analyze the query complexity of algorithms and subroutines, i.e., the number of queries to controlled versions of the input oracles and their inverses. We will denote the optimal quantum query complexity of an $\eps$-approximate \emph{quantum SDP-solver} with success probability $2/3$ by $T_{SDP}(\eps)$. We only consider success probability $2/3$ to simplify the notation and proofs. However in all cases an $\eps$-approximate SDP-solver with success probability $1-\zeta$ can easily be constructed using $\bigO{\log(1/\zeta) T_{SDP}(\eps)}$ queries.

In our algorithms we will assume access to a quantum-read/classical-write RAM (QCRAM), and assume one read/write operation has a constant gate complexity\footnote{Note that read/write operations of a QRAM or QCRAM of size $S$ can be implemented using $\bOt{S}$ two-qubit gates, so this assumption could hide a factor in the gate complexity which is at most $\bOt{S}$.}; the size of the QCRAM will typically be $\BOT{n,m}{\left(\frac{Rr}{\eps}\right)^{\!2}}$ bits. Most often in our results the number of non-query elementary operations, i.e., two-qubit gates and QCRAM calls, matches the query complexity up to polylog factors. In particular, if not otherwise stated, in our results a $T$-query quantum algorithm uses at most $\BOT{n,m}{T}$ elementary operations. 

\paragraph{Subroutines.}
We will work with two major subroutines which need to be implemented according to the specific input model. First, the algorithm will require an implementation of a Gibbs-sampler.
\begin{definition}[Gibbs-sampler]
  A $\theta$-precise \emph{Gibbs-sampler} is a unitary that takes as input a data structure storing a vector $y\in \mathbb{R}^{m+1}_{\geq 0}$ and creates as output a purification of a $\theta$-approximation in trace distance of the Gibbs state $e^{-\sum_{j=0}^m y_jA_j } / \tr{e^{-\sum_{j=0}^m y_jA_j}}$.
	If $\nrm{y}_1 \leq K$ and the support of $y$ has size at most~$d$, then we write $T_{Gibbs}(K,d,4\theta)$ for the cost of this unitary.
  
For technical reasons we also allow Gibbs-samplers that require a random classical input seed $S \in \{0,1\}^a$ for some $a=\bigO{\log(1/\theta)}$. In this case the output should be a $\theta$-approximation of the Gibbs state with high probability $(\geq 4/5)$ over a uniformly random input seed $S$.
\end{definition}

We will use the approximate Gibbs states in order to compute the quantity $\tr{A_j\rho}$ using a trace estimator.
\begin{definition}[Trace estimator]\label{def:traceEstimator}
  A $(\theta,\sigma)$-\emph{trace estimator} is a unitary that as input takes a state $\rho$ and index $j$. It outputs a sample from a random variable $X_j\in \mathbb{R}$ such that $X_j$ is a trace estimator that is at most $\theta/4$ biased:
\[
  |\tr{A_j\rho} - \mathbb{E}[X_j]| \leq \theta / 4,
\]
and the standard deviation of $X_j$ is at most $\sigma$. 
  We write $T_{Tr}^{\sigma}(\theta)$ for the cost of such a unitary.
\end{definition}

\subsection{Previous work}
Classical SDP-solvers roughly fall into two categories: those with logarithmic dependence on $R$, $r$ and $1/\eps$, and those with polynomial dependence on these parameters but better dependence on $m$ and $n$.  In the first category the best known algorithm~\cite{lsw:faster} at the time of writing has complexity 
$$
\BOT{Rr/\eps}{m(m^2+n^\omega + mns)}.
$$
where $\omega \in [2,2.38]$ is the yet unknown exponent of matrix multiplication.

In the second category Arora and Kale~\cite{arora&kale:sdp} gave an alternative framework for solving SDPs, using a matrix version of the ``multiplicative weights update'' method. Their framework can be tuned for specific types of SDPs, allowing for near linear-time algorithms in the case of for example the Goemans-Williamson SDP for the approximation of the maximum cut in a graph~\cite{GoemansWilliamson95}.

In 2016 Brand\~ao and Svore~\cite{brandao:quantumsdp} used the Arora-Kale framework to implement a general quantum SDP-solver in the sparse matrix model. They observed that the matrix
\[
\rho := \frac{e^{-\sum_{j=0}^m y_jA_j}}{\tr{e^{-\sum_{j=0}^m y_jA_j}}},
\]  
that is used for calculations in the Arora-Kale framework is in fact a $\log(n)$-qubit Gibbs state and can be efficiently prepared as a quantum state on a quantum computer. Using this they achieved a quantum speedup in terms of $n$. Combining this with a Grover-like speedup allowed for a speedup in terms of $m$ as well, leading to an $\eps$-approximate quantum SDP solver with complexity
\[
\bOt{\sqrt{mn}s^2 \left(\frac{Rr}{\eps}\right)^{\!\!32}}.
\]
They also showed an $\Omega(\sqrt{m}+\sqrt{n})$ quantum query lower bound for solving SDPs when all other parameters are constant. This left as open question whether a better lower bound, matching the $\sqrt{mn}$ upper bound, could be found.
The upper bound for the sparse input model was subsequently improved by van Apeldoorn et al.~\cite{AGGW:SDP} to
\[
\bOt{\sqrt{mn}s^2 \left(\frac{Rr}{\eps}\right)^{\!\!8}}.
\] 
van Apeldoorn et al.~also gave an $\Omega(\sqrt{\max(n,m)}\min(n,m)^{3/2})$ lower bound, albeit for non-constant parameters $R$ and $r$. This bound implies that there is no general quantum SDP-solver that has a $o(nm)$ dependence on $n$ and $m$ and logarithmic dependence on $R$, $r$ and $1/\eps$. They also showed that every SDP-solver whose efficiency relies on outputting sparse dual solutions (including their algorithm and that of Brand\~ao and Svore~\cite{brandao:quantumsdp}) is limited, since problems with a lot of symmetry (like maxflow-mincut) in general require non-sparse dual solutions. Furthermore, they showed that for many combinatorial problems (like MAXCUT) $R$ and $r$ increase linearly with $n$ and $m$.  

Very recently Brand\~ao et al.~\cite{brandao:expSDP} gave an improved SDP-solver for the quantum state input model\footnote{This model was already introduced in the first version of~\cite{brandao:quantumsdp} together with a similar complexity statement, but there were some unresolved issues in the proof, that were only fixed by the contributions of~\cite{brandao:expSDP}.} that has a complexity bound with logarithmic dependence on $n$:
\[
T_{SDP}(\eps)=\BOT{n}{\sqrt{m}\ \poly\left(\frac{Rr}{\eps},B,\max_{j\in \{0,\dots, m\}}[\mathrm{rank}(A_j)]\right)}.
\] 
Brand\~ao et al.~also applied their algorithm to the problem of \emph{shadow tomography}, giving the first non-trivial application of a quantum SDP-solver. 

Subsequently these results where further improved by the introduction of the Fast Quantum OR lemma by the same authors~\cite{brandao:expSDP2}. Approaches prior to~\cite{brandao:expSDP2} searched for a violated constraint in the SDP using Grover-like techniques, resulting in a multiplicative complexity of Gibbs-sampling and searching. The Fast Quantum OR lemma can be used to separate the search phase from the initial Gibbs-state preparation phase. This led to the improved complexity bound~\cite{brandao:expSDP2} of
\[
\BOT{n}{\left( \sqrt{m}+\poly(\max_{j\in \{0\ldots m\}}[\mathrm{rank}(A_j)]) \right) \ \poly\left(\frac{Rr}{\eps},B\right)}.
\]
Using the Fast Quantum OR Lemma the complexity bound on $T_{SDP}(\eps)$ can be improved in the sparse input model as well,
as independently observed by the authors of~\cite{brandao:expSDP2} and by us. We thank the authors of~\cite{brandao:expSDP2} for sending us an early draft of~\cite{brandao:expSDP2} introducing the Fast Quantum OR Lemma, which enabled us to work on these improvements. During the correspondence the application of the OR lemma to the sparse matrix model was independently suggested by Brand\~ao et al.~\cite{wu:email} and by us.

\subsection{Our results}
In this paper we build on the Arora-Kale framework for SDP-solving in a similar fashion as \cite{AGGW:SDP,brandao:quantumsdp} and also use results from \cite{brandao:expSDP,LRS15} to construct a primal oracle. We improve on the previous results about quantum SDP-solving in three different ways:
\begin{itemize}
\item We give a computationally more efficient version of the Gentle Quantum Search Lemma~\cite{aaronson:shadow} using the Fast Quantum OR Lemma from \cite{brandao:expSDP2}. We also extend this to minimum finding to get our \emph{Two-Phase Quantum Minimum finding} (Lemma~\ref{lem:minfinding}). As independently observed by the authors of \cite{brandao:expSDP2} the Fast Quantum OR Lemma gives a speed-up for SDP primal oracles in general. Moreover, using Two-Phase Quantum Minimum finding, we show how to improve the upper bound on the complexity of general SDP-solving from
\begin{equation}\label{eq:SDPCostAmpEst}
T_{SDP}(\eps) = \BOT{n}{\sqrt{m}  \left(T_{Tr}^{\sigma}(\gamma)T_{Gibbs}(\gamma,\gamma^2,\gamma^{-1})\right)  \gamma^3\sigma}
\end{equation}
as implied in previous work~\cite{brandao:quantumsdp,AGGW:SDP} to
\begin{equation}\label{eq:SDPCostOR}
T_{SDP}(\eps) = \BOT{n}{  \left( \sqrt{m}  T_{Tr}^{\sigma}(\gamma) + T_{Gibbs}(\gamma,\gamma^2,\gamma^{-1})  \right) \gamma^4\sigma^2},
\end{equation}
where $\gamma = \Theta( Rr/\eps)$. For the complexity of SDP primal oracles, the same upper bounds holds. \anote{Removed the constants from the defintions in order to improve readability. }

\item We introduce the quantum operator input model, and show that it is a simultaneous generalization of the other two input models which were considered earlier. In particular we show that both the sparse model and the quantum state mdoel can be reduced to the quantum operator model model with a constant overhead and with the choices of $\alpha= s$ and $\alpha= B$ respectively. Moreover, we show that for $\sigma=\Theta(1)$, we have that 
$$ T_{Tr}^{\sigma}(\gamma)=\BOT{\gamma}{\alpha},$$
in the quantum operator model.
We also show how to simulate a linear combination of Hamiltonians efficiently using this input model, and prove that 
\[
T_{Gibbs}(K,d,\theta) = \BOT{\theta,d}{\sqrt{n}K\alpha}.
\] 
This result is based on the idea of gradually building up an efficient data structure for state preparation, following ideas of~\cite{KerenidisQuantumRecommendation}. This significantly improves the complexity of Gibbs-sampling compared to \cite{AGGW:SDP}, 
which presented a Gibbs sampler subroutine in the sparse matrix input model with complexity
\[
T_{Gibbs}(K,d,\theta) = \BOT{\theta}{\sqrt{n}Ks^2d^2}.
\]

\item We develop a new method for Gibbs-sampling in the quantum state model. Our approach, in contrast to the one in \cite{brandao:expSDP2}, does not introduce a dependence on the rank of the input matrices in the complexity. In particular we improve the complexity bound of~\cite{brandao:expSDP2} 
\[
T_{Gibbs}(K,d,\theta) = \bigO{\poly(K,B,d,1/\theta,\max_{j\in \{0\ldots m\}}[\mathrm{rank}(A_j)])}
\]
to 
\[
T_{Gibbs}(K,d,\theta) = \BOT{d,\theta,n}{ (KB)^{3.5} }.
\]
An important consequence of this improvement is that in the complexity of SDP solving we do not get a dependence on the rank of the input matrices, unlike Brand\~ao et al.~\cite{brandao:expSDP2}. For some quantum SDPs given in the quantum state input model the $A_j$ matrices could naturally correspond to quantum states. In this case $B$ would be just $1$, but the rank could easily be proportional to $n$, e.g., for a highly mixed state, eliminating the speedup over the sparse input model. Finally note that this Gibbs-sampling method is only beneficial if\nolinebreak $\sqrt{n}\leq (KB)^{2.5}$, otherwise the reduction to the quantum operator model with $\alpha = B$ gives a better algorithm.
\end{itemize}

\noindent For the quantum operator input model the above improvements lead to the complexity bound 
\begin{equation}\label{eq:sqrtnmImprovement}
T_{SDP}(\eps) = \bOt{  \left( \sqrt{m}  +  \sqrt{n}\gamma  \right)\alpha \gamma^4},
\end{equation}
where $\gamma:= \frac{Rr}{\eps}$. Note that the $\Omega(\sqrt{n}+\sqrt{m})$ lower bound of~\cite{brandao:quantumsdp} also applies to the quantum operator model due to our reductions, matching the above upper bound \eqref{eq:sqrtnmImprovement} up to polylog factors in $n$ and $m$ when $\gamma$ and $\alpha$ are constant.
For the quantum state input model our improved Gibbs-sampler yields the complexity bound
\[
T_{SDP}(\eps) = \BOT{n}{ \left(\sqrt{m}+B^{2.5}\gamma^{3.5}\right)B\gamma^4}.
\]
In both cases, the same bound holds for an SDP primal oracle but with $\gamma:=R/\eps$.

\begin{savenotes}
\begin{table}[H]
	\centering
	\label{my-label}
	\fontsize{8pt}{8pt}\selectfont
	\def\arraystretch{2}

	\begin{tabular}{c|c|c|c|c|}\cline{2-5}
		& \multicolumn{2}{|c}{With OR lemma / Two-Phase Search} & \multicolumn{2}{|c|}{Without OR lemma / Two-Phase Search} \\ \cline{2-5}
		& Sparse input & Quantum state input & Sparse input & Quantum state input \\ \hline
		\multicolumn{1}{|c|}{Previous}
		& $\bOt{\left(\sqrt{m}+\sqrt{n}s\gamma^{5}\right)s \gamma^{4}}$ 
		&  \kern-1mm $\BOT{\!n\!}{ \left(\sqrt{m}+\poly(\mathrm{rk})\right)\poly\left(\gamma,B\right)}$ \kern-2mm 
		& $\bOt{\sqrt{mn}s^2 \gamma^{8}}$ 
		& $\BOT{\!n\!}{ \sqrt{m}\poly\left(\gamma,B,\mathrm{rk}\right)}$\\
		\multicolumn{1}{|c|}{Gibbs-sampling}
		& Theorem\footnote{\label{noteIndep}A similar result was independently proved by Brand\~ao et al.~\cite{brandao:expSDP2}.}~\ref{thm:mainSDP} + \cite{AGGW:SDP} 
		&  \cite{brandao:expSDP2} & \cite{AGGW:SDP}            &  \cite{brandao:expSDP}    \\ \hline     
		\multicolumn{1}{|c|}{Improved}
		& $\bOt{\left(\sqrt{m}+\sqrt{n}\gamma\right)s \gamma^{4}}$
		& $\BOT{\!n\!}{ (\sqrt{m} + B^{2.5}\gamma^{3.5})B\gamma^4}$
		& $\bOt{\sqrt{mn}s \gamma^{4}}$     
		&  $\BOT{\!n\!}{ \sqrt{m}B^{3.5}\gamma^{6.5}}$  \\ 
		\multicolumn{1}{|c|}{Gibbs-sampling}
		&  Theorem~\ref{thm:simpGibbs}          &  Theorem~\ref{thm:difGibbs}
		&  Corollary~\ref{cor:simpGibbsNoOR}           &  Corollary~\ref{cor:difGibbsNoOR}  \\ \hline          
	\end{tabular}
	\caption{Summary of our query complexity bounds illustrating the role of our various improvements.
		Here we present the results for the sparse matrix and quantum state input models for comparison to prior work. However, note that our results presented for sparse input hold more generally for the quantum operator input model; to get the corresponding results one should just replace $s$ by $\alpha$ in the table. Thereby similar bounds hold in the case of the quantum state input model too, after replacing $s$ by $B$, which can be beneficial when $B^{2.5}\gamma^{2.5}\geq \sqrt{n}$.  Notation: $\mathrm{rk}=\max_{j\in\{0,\dots, m\}} \text{rank}(A_j)$ and $\gamma=\frac{Rr}{\eps}$.}	
\end{table}
\end{savenotes}

\noindent In Section~\ref{sec:applications} we give some applications of quantum SDP-solvers:
\begin{itemize}
\item We extend the idea of applying SDP-solving to the problem of shadow tomography: given an unknown, $n$-dimensional quantum state $\rho$, find $\eps$-additive approximations of the expectation values $\tr{E_1\rho},\dots,\tr{E_m\rho}$ of several binary measurement operators. This problem was introduced by Aaronson in \cite{aaronson:shadow}, he gave an efficient algorithm in terms of the number of samples from $\rho$. In particular he proved that $\bOt{\log^4(m)\log(n) / \eps^5}$ samples suffice. Brand\~ao et al.~\cite{brandao:expSDP2} applied their SDP-solver to get a more efficient algorithm in terms of computation time when the measurements $E_i$ are given in the quantum state model, while keeping the sample complexity as low as $\poly(\log(m),\log(n),1/\eps,B)$. We simultaneously improve on both results, giving a sample bound of $\bOt{\log^4(m)\log(n) / \eps^4}$ while also improving the best known time complexity~\cite{aaronson:shadow,brandao:expSDP2} of the implementation for all input models.
Finally we show that if we can efficiently implement the measurements $\tr{E_1\rho},\dots,\tr{E_m\rho}$ on a quantum computer, then we can also efficiently represent $E_1,\dots,E_m$ using the quantum operator input model.

\item We apply the SDP-solvers to the problem of Quantum State discrimination: given a set of quantum states, what is the best POVM for discriminating between the states? We consider the case of minimizing the total error in the measurements. In this case we get an algorithm with running time $\bOt{\sqrt{k}\ \poly(d,1/\eps)}$ in the sparse input model, where $k$ is the number of states and $d$ is the dimension of the states. Due to the quantum state model for SDP-solving, we can also solve the problem when the states that need to be discriminated are actually given as quantum states, rather than classical descriptions of density operators.

\item We use the SDP-solver for the sparse matrix model to solve the problem of E-optimal design: given a set of $k$ experiments, find the optimal distribution of the experiments that minimizes the variance in our knowledge of a $d$-dimensional system. Our final bound is $\bOt{(\sqrt{k}+\sqrt{d})\poly(1/\eps,P)}$, where $P$ is a parameter that depends on the standard deviation of the experiments.
\end{itemize}

We end the paper with proving new lower bounds. Lower bounds on the quantum query complexity of SDP-solving for the sparse input model were presented in previous works~\cite{brandao:quantumsdp,AGGW:SDP}. We add to this by giving $\Omega(\sqrt{m}B/\eps)$ and $\Omega(\sqrt{m}\alpha/\eps)$ bounds for the quantum state model and quantum operator model respectively. These lower bounds show that the $\sqrt{m}$ factor and the polynomial dependence on the parameters $B,\alpha$, and $1/\eps$ are necessary. 

Compared to problems with a discrete input, proving lower bounds on continuous-input quantum problems gives rise to extra challenges and often requires more involved techniques, see for example the work of Belovs~\cite{belovsGeneralAdv15} on generalizations of the adversary method. 
Due to these difficulties, fewer results are known in this regime. Examples of known continuous-input lower-bound results include phase-estimation related problems (cf.~Bessen~\cite{bessenlowerbound}) and the complexity-theoretic version of the no-cloning theorem due to Aaronson~\cite{aaronson:quantummoney}.
Recently, a new hybrid-method based approach was developed by Gilyén et al.~\cite{gilyen:qgradient} in order to handle continuous-input oracles, which they use for proving a lower bound for gradient computation. We use their techniques to prove our lower bounds, combined with efficient reductions between input models stemming from the smooth-functions of Hamitonians techniques developed in the work of van Apeldoorn et al.~\cite{AGGW:SDP}.

\section{SDP-solving frameworks}
In this section we present two frameworks for SDP-solving. The first is the Arora-Kale framework, which is used to find a good approximation of the optimal value and an almost feasible solution to the dual.
Then we present an algorithm to implement a primal oracle. These together implement a full SDP-solver.

\subsection{The Arora-Kale framework}

Similarly to previous work \cite{brandao:quantumsdp,AGGW:SDP} we build our results on the Arora-Kale framework. For a detailed description see the original paper by Arora and Kale \cite{arora&kale:sdp}. For the specific application to general SDP-solvers, see \cite{AGGW:SDP}. For our application, the following broad overview suffices. 

We assume that the first constraint is $\tr{X}\leq R$, i.e., $A_1=I$ and $b_1=R$. Remember that we set $A_0 = -C$ and $b_0 = -g$.

\begin{enumerate}
	\item let $y = 0 \in \mathbb{R}^{m+1}$ and set $\theta = \frac{\eps}{6Rr}$.
	\item Repeat $\frac{\ln(n)}{\theta^2}$ times the following:
	\begin{enumerate}
		\item Define $\rho := e^{-\sum_{j=0}^m y_jA_j} / \tr{ e^{-\sum_{j=0}^m y_jA_j} }$.
		\item Find a $\tilde{y}$ in the polytope
		\begin{align*}
		\mathcal{P}_{\delta}(\rho) := \Big\{ \tilde{y} \in \mathbb R^{m+1}:\ &b^T \tilde{y}  \leq 0, \\
		&\sum_{j=0}^m \tilde{y}_j\tr{ A_j  \rho} \geq -\delta,\\ 
		& \tilde{y} \geq 0, \ \tilde{y}_0=\frac{1}{2r}, \ \nrm{\tilde{y}}_1 \leq 1 \Big\}.
		\end{align*}
		for $\delta = \theta$ or conclude that none exists for $\delta = 0$.
		\item If no such $\tilde{y}$ exists, then conclude that $\opt > g$ and stop.
		\item If such a $\tilde{y}$ exists, then update $y \leftarrow \theta\tilde{y}$.
	\end{enumerate}
	\item Conclude $\opt \leq g+\eps$ and output $ \frac{2r\theta}{\ln(n)} y + \frac{\eps}{R} e_1 - e_0$ as a dual solution.
\end{enumerate}

Brand\~ao and Svore~\cite{brandao:quantumsdp} observed that $\rho := e^{- \sum_j y_jA_j} / \tr{ e^{- \sum_j y_jA_j} }$ is a quantum Gibbs state and this state can be prepared efficiently on a quantum computer, allowing fast trace estimation, in particular resulting in a quadratic speedup in $n$.  

A procedure that solves step (b) is called a $\theta$-oracle. In the rest of this paper we will assume that the cost of updating the $y$ vector is lower than the cost of a $\theta$-oracle call. 

Van Apeldoorn at al.\cite{AGGW:SDP} gave an oracle implementation that always outputs a $2$-sparse $\tilde{y}$. The oracle is constructed using a geometric argument that boils down to minimizing over $m$ angles, one for each constrained. Each angle is easily computed from the corresponding $b_j$ and $\tr{A_j \rho}$, where a $\theta$-additive error is allowed in the approximation of $\tr{A_j\rho}$. In \cite{AGGW:SDP} this minimization was done using Quantum Minimum finding \cite{durr&hoyer:minimum}, allowing for a quadratic speedup in $m$. Previously Brand\~ao and Svore \cite{brandao:quantumsdp} applied other techniques to similarly get a quadratic speedup in $m$ but this introduced a worse dependence on $Rr/\eps$.

\subsection{An SDP primal oracle}
To construct a primal oracle, we use the same algorithm as Brand\~ao et al.\cite{brandao:expSDP2} following the proof of Lemma 4.6 of Lee, Raghavendra and Steurer~\cite{LRS15}. 
A few small reductions are required to apply this technique.
To be able to work with density operators instead of $X$, the $b_j$s in the constraints $1\ldots m$ are scaled down by a factor $R$, such that every solution $X'$ to the new SDP has trace at most $1$.
Then, we add one new variable denoted by $\omega$ such that
\[
\rho :=\begin{bmatrix}
X' & 0 \\
0 & \omega
\end{bmatrix}.
\]
Now $\tr{\rho} = 1$ and $\rho \succeq 0$ imply that $\tr{X'}\leq 1$, and we get a new SDP that is equivalent to the previous one. It can be shown that in our input models this reduction does not introduce more than a constant factor overhead in the complexity.

The framework for an SDP primal oracle can now be summarized as follows (here we write $A_j$ and $b_j$ for the input after the reductions).
\begin{enumerate}
\item Let $y = 0 \in \mathbb{R}^{m+1}$ and $\theta = \frac{\eps}{2R}$.
\item Repeat $\frac{\ln(n)}{\theta^2}$ times the following:
\begin{enumerate}
\item Define $\rho := e^{-\sum_{j=0}^m y_jA_j} / \tr{ e^{-\sum_{j=0}^m y_jA_j} }$.
\item Find an index $j$ such that $\tr{A_j \rho} \geq b_j$ or conclude correctly that for all $j$, $\tr{A_j \rho} \leq b_j+\theta$.
\item If no $j$ is found, then we are done and output $y$ and $z = R\tr{X'}$, where $\tr{X'}$ is the probability\footnote{Note that a $\theta$-approximation of $\tr{X'}$ is easy to compute by means of amplitude estimation if $\rho$ can be efficiently prepared as a quantum state -- which is the case in our algorithms.} of measuring $\rho$ to be in the subspace corresponding to the variable $X'$.
\item Otherwise update $y \leftarrow y+\theta e_j$.
\end{enumerate}
\item Conclude that there is no solution for $\theta = 0$.
\end{enumerate}

Both frameworks have a very similar structure. The main difference is that the primal oracle framework requires only a simple search, whereas the $\theta$-oracle needed for the Arora-Kale framework is slightly more complex. Our implementation of the $\theta$-oracle is always returning a $3$-sparse vector, thus in both cases we will work with a $y$ vector that is non-negative and $\BOT{n}{1/\theta^2}$ sparse.
\section{Improvements}

\subsection{Fast Quantum OR Lemma and Two-Phase Minimum Finding}

To speed up the SDP solvers derived form this framework we use a fast version of the Quantum OR Lemma, as in \cite{brandao:expSDP2}. They prove the following lemma:
\begin{lemma}[Fast Quantum OR Lemma~\cite{brandao:expSDP2}]\label{lem:qOR}
	Let $\Pi_1,\dots,\Pi_m$ be projectors and $\rho$ a quantum state. Suppose that either
	\begin{enumerate}
		\item $\exists j$ s.t. $\tr{\Pi_j\rho}\geq 1-\delta_1$, or
		\item $\frac{1}{m}\sum_{j=1}^m \tr{\Pi_j\rho}\leq \delta_2$
	\end{enumerate}
	for some $0 < \delta_1 \leq 1/2$ and $0< \delta_2 \leq \frac{(1-\delta_1)^2}{12m}$. Then for all $\xi \in(0,1)$ there is a procedure that accepts with probability at least $(1-\delta_1)^2/4-\xi$ in the first case, and probability at most $3m\delta_2 + \xi$ in the second case, and that uses $1$ copy of $\rho$ and $\bigO{\log(m)\sqrt{m}/\xi}$ applications of a controlled version of the reflection $I-2\sum_{j=1}^m\Pi_j\otimes \ketbra{j}{j}$.
\end{lemma}

This lemma is almost the same as the original Quantum OR Lemma~\cite{harrow:orbound} but with the addition that the algorithm requires only $\bigO{\log(m)\sqrt{m}/\xi}$ applications of the controlled reflection.
In a recent paper Aaronson~\cite{aaronson:shadow} proved the \emph{Gentle Quantum Search Lemma} using the Quantum OR Lemma. His proof can easily be extended to use the Fast Quantum OR Lemma. We call the resulting more efficient version \emph{Two-Phase Quantum Search}. 

In the setting of the Two-Phase Quantum Search we will have $m$ algorithms for decision problems and we ask whether one of them evaluates to $1$, and if so, to find one. We also know that all algorithms start with preparing some state $\rho$, followed by some procedure $U_j$ that depends on the index of the decision problem $j\in[m]$. In classical deterministic processes it is quite natural that only one preparation of $\rho$ is needed since the result can be stored. For bounded error classical processes $\bigO{\log(m)}$ preparations of $\rho$ suffice to get the error probability of one decision problem below $1/m$. By the classical union bound this is low enough that we can find a marked element with constant probability. However, if $\rho$ is a quantum state and the $U_j$ are quantum algorithms, then such a bound is not so straightforward, since progress made in constructing $\rho$ might be destroyed when running one of the $U_j$. Nevertheless, using the Fast Quantum OR Lemma it can be shown that $\bOt{\log^4(m)}$ samples from $\rho$ suffice. 

\begin{lemma}[Two-Phase Quantum Search]\label{lem:tpqs}
Let $\nu\in(0,1)$. Let $\rho$ be a quantum state and $U_1,\dots,U_m$ be unitaries with the $U_j$ accessible through a unitary $U$ that acts as $U\ket{j}\ket{\psi} = \ket{j}U_j\ket{\psi}$.
	Then there is a quantum algorithm that using $\bOt{\log^4(m)\log(\nu)}$ samples of $\rho$ and $\bOt{\sqrt{m}\log(\nu)}$ applications of $U$ and its inverse, outputs with success probability at least $1-\nu$ either
	\begin{itemize}
		\item a $j$ such that $\tr{(I\otimes\ketbra{1}{1}) U_j\rho U_j^{\dagger}}\geq 1/3$, i.e., a $j$ such that $U_j$ outputs $1$ with probability at least $1/3$ on input $\rho$,
		\item or concludes correctly that $\tr{(I\otimes\ketbra{1}{1}) U_j\rho U_j^{\dagger}} < 2/3$  for all $j$, i.e., no unitary outputs $1$ with probability at least $2/3$ on input $\rho$.
	\end{itemize}
\end{lemma}
\begin{proof}
This follows from the proof of Gentle Quantum Search in~\cite[Lemma 15]{aaronson:shadow} using the Fast Quantum OR Lemma~\cite{brandao:expSDP2} instead the normal Quantum OR Lemma.
\end{proof}

Using the above lemma we construct the \emph{Two-Phase Quantum Minimum Finding} algorithm.
It turns out that we need to use this algorithm in a situation where different values have different error-bars, therefore the statement gets slightly complicated. In typical use-cases one can probably just choose each error-margin $\eta$ equal to say $\delta$ resulting in a simpler statement. 

\begin{lemma}[Two-Phase Quantum Minimum Finding]\label{lem:minfinding}
Let $\delta,\nu' \in (0,1)$. Let $\rho$ be a quantum state and $U_1,\dots,U_m$ be unitaries, with the $U_j$ accessible through a unitary $U$ that acts as $U\ket{j}\ket{\psi} = \ket{j}U_j\ket{\psi}$.
	Let $a_1,\dots,a_m$, $\eta_1,\dots,\eta_m$ be numbers such that $\min_j |a_j|+|\eta_j| \leq M$. Assume that with probability at least $2/3$, $U_j$ computes a binary representation of $a_j$ up to additive error $\eta_j$ using one copy of $\rho$.
	Then, with probability at least $1-\nu'$, we can find a $j$ such that $a_j-\eta_j\leq \min_i (a_i+\eta_i) +\delta$ using $\bOt{\log^4(m)\log(M/\delta)\log(\nu')}$ samples of $\rho$ and $\bOt{\sqrt{m}\log(\nu') \log(M/\delta)}$ applications of $U$ and its inverse.
\end{lemma}
\begin{proof}
	Do a binary search on the value $v$ to precision $\delta$ by checking whether there is still an element with $a_i+\eta_i\leq v$ using Lemma~\ref{lem:tpqs} in each round with setting $\nu=\Theta(\nu'/\log(M/\delta))$. This binary search will result in a value $v\leq\min_i (a_i+\eta_i) +\delta$ with probability at least $1-\nu/2$, and it is not hard to see that the last $j$ found by Lemma~\ref{lem:tpqs} during the binary search will be such that $a_j-\eta_j\leq v$ with probability at least $1-\nu/2$. Therefore this $j$ satisfies the required inequality with probability at least $1-\nu$.
\end{proof}
This leads to the following general bound on SDP-solving.
\begin{theorem}\label{thm:mainSDP}
  Assume that updating an entry of $y \in \mathbb{R}^{m+1}$ in the data structure requires at most  $\bOt{T_{Gibbs}(\gamma,\gamma^2,\gamma^{-1})}$ elementary operation, where $\gamma := 6R r/\eps$. Then there is a quantum SDP-solver for which
	\[
          T_{SDP}(\eps) = \BOT{n}{  \left( \sqrt{m}  T_{Tr}^{\sigma}(\gamma^{-1}) + T_{Gibbs}(\gamma,\gamma^2,\gamma^{-1})  \right) \gamma^4\sigma^2},
	\]
	 similarly there is also a quantum algorithm with the same complexity, but with  $\gamma := 6R/\eps$, that implements an SDP primal oracle.
\end{theorem}
\begin{proof}
To construct an SDP-solver use both frameworks in succession, otherwise use only the primal oracle.
The frameworks run for $\BOT{n}{\gamma^2}$ iterations. In each iteration we need to update at most three entries of the $y$ vector, which takes at most $\bOt{T_{Gibbs}(\gamma,\gamma^2,\gamma^{-1})}$ elementary operations by assumption. To search for a violated constraint when using the primal oracle framework, we use the Two-Phase Quantum Search, and we use Two-Phase Minimum Finding to implement the minimum finding needed in the Oracle for the Arora-Kale framework\footnote{In the Oracle implementation of~\cite[Lemma~14]{AGGW:SDP} the minimum finding is not done over the computed traces, but rather the angles calculated using these traces. The trace $\rightarrow$ angle conversion suffices with precision $\delta^{-1}=\bigO{\poly(\gamma)}$, and since the magnitude $M$ of angles is bounded by $\pi$, we get $\log(M/\delta)=\bigO{\log(\gamma)}$ in Lemma~\ref{lem:minfinding}.}, following the geometric approach of~\cite[Lemma~14]{AGGW:SDP} for implementing a $\gamma^{-1}$-oracle.

Let $\rho:=\rho_{S}^{\otimes k}$ where $k=6(4 \sigma \gamma)^2$, $S$ is a uniform random seed and $\tilde{\rho}_S$ is a density operator, that is a $\gamma^{-1}/4$-approximation in the trace distance of the Gibbs state $\rho_{Gibbs(y)}$ corresponding to the current $y$ vector (for at least a $4/5$ fraction of the possible input seeds). Let $U_j$ be the operator that applies a $(\gamma^{-1}/4,\sigma)$-trace estimator to each copy of $\tilde{\rho}_S$ and takes the average of the outcomes. I.e., it obtains estimates of $\tr{A_j\tilde{\rho}_S}$ with bias at most $\gamma^{-1}/4$ and standard deviation at most $\sigma$ independently $k$-times, taking the average at the end. By Chebyshev's inequality we can see that this way $U_j$ computes a $2\gamma^{-1}/4$-precise estimate of $\tr{A_j\tilde{\rho}_S}$ with probability at least $5/6$.
Also with probability at least $4/5$ we have that $\nrm{\tilde{\rho}_S-\rho_{Gibbs(y)}}\leq \gamma^{-1}/4$ and thus we get a $3\gamma^{-1}/4$-precise estimate of $\tr{A_j\rho_{Gibbs(y)}}$ with probability at least $4/5\cdot 5/6=2/3$.
The preparation of $\rho$ can be performed using $k T_{Gibbs}(\gamma,\gamma^2,\gamma^{-1})$ queries by definition, whereas $U_j$ can be implemented with query complexity $k T_{Tr}^{\sigma}(\gamma^{-1})$ .

By using Two-Phase Quantum Search and Two-Phase Quantum Minimum Finding with $\nu = \BOT{n}{\gamma^{-2}}$ we get that it takes $\BOT{n}{\left(( \sqrt{m}  T_{Tr}^{\sigma}(\gamma^{-1}) + T_{Gibbs}(\gamma,\gamma^2,\gamma^{-1})\right)\gamma^2\sigma^2}$ queries to implement an iteration. The stated final complexity follows considering the number of iterations.
\end{proof}
In the rest of this paper we will give upper bounds on $T_{Gibbs}$ and $T_{Tr}^{\sigma}$ for different input models. In particular, due to the results from the next section $T_{Gibbs}(K,d,\theta)$ will always depend only logarithmically on $d$ and $\theta$, and $T_{Tr}^{\sigma}(\theta)$ will only depend logarithmically on $\theta$. Nevertheless we left these parameters in Theorem~\ref{thm:mainSDP} for completeness and to allow for comparison with previous results.

\subsection{SDP-solving using the quantum operator input model}\label{subsec:Hamsim}
In this section we present recent results on Hamiltonian simulation which motivate the quantum operator input model. In particular, we restate results from \cite{LowChuang:Qubitization2016} showing that using ``block-encodings'' of Hamiltonians, efficient Hamiltonian simulation can be implemented, moreover we show that this input model generalizes both other input models. Then we show that efficient Hamiltonian simulation leads to efficient trace estimation. Since we need to prepare Gibbs states for Hamiltonians of the form $\sum_{j=0}^m y_jA_j$ we also show how to obtain a block-encoding of a linear combination of Hamiltonians using an efficient data structure. Combining these techniques with the meta-algorithm of Theorem~\ref{thm:mainSDP} leads to an efficient quantum SDP-slover, see Theorem~\ref{thm:simpGibbs}.

\subsubsection{Block-encodings of operators}

Inspired by the work of Low and Chuang~\cite{LowChuang:Qubitization2016} and following the exposition of Gilyén and Wiebe~\cite{gilyenBlockMatrices} we introduce block-encodings of operators, which will be the central concept to the quantum operator approach. Recall that a unitary $U$ is a block-encoding of a matrix $A$ if the top-left block of $U$ is proportional to $A$:
$$ U=\left(\begin{array}{cc} A/\alpha & . \\ . & .\end{array}\right) $$
Note that $\nrm{U}=1$ hence we must have $\nrm{A}\leq \alpha$.
\paragraph{Notation.} Since unitaries are also operators, we will sometimes call a unitary a $(1,a,\eps)$-block-encoding of another unitary if it uses $a$ ancillary qubits and is an $\eps$-approximation in the operator norm. Note that every unitary is a $(1,0,0)$-block-encoding of itself.
\\\\
\noindent The following Hamiltonian simulation theorem is a corollary of the results of \cite[Theorem~1]{LowChuang:Qubitization2016}, and provides the main motivation for this input model. For a detailed proof see the work of Chakraborty, Gilyén and Jeffrey~\cite{CGJ:PowerOfBlockPowers18}.
\begin{theorem}[Optimal block-Hamiltonian simulation \cite{LowChuang:Qubitization2016}] \label{thm:blockHamSim}
	Suppose that $U$ is an $(\alpha,a,\eps/|2t|)$-block-encoding of the Hamiltonian $H$. Then we can implement an $\eps$-precise Hamiltonian simulation unitary $V$ which is an $(1,a+2,\eps)$-block-encoding of $e^{itH}$, with $\bigO{|\alpha t|+\frac{\log(1/\eps)}{\log\log(1/\eps)}}$ uses of controlled-$U$ or its inverse and with $\bigO{a|\alpha t|+a\frac{\log(1/\eps)}{\log\log(1/\eps)}}$ two-qubit gates.
\end{theorem}

Following the approach of~\cite{AGGW:SDP}, we show in Appendix~\ref{apx:HamSim} that efficient Hamiltonian simulation implies efficient trace estimation:

\begin{corollary}\label{cor:specTraceCalc}
	For $\sigma = 6$, we have that $T_{Tr}^{\sigma}(\theta)=\BOT{\frac{n}{\eps}}{\alpha}$ in the quantum operator input model.
\end{corollary}
\begin{proof}
	This is shown by Corollary~\ref{cor:traceCalc} from Appendix~\ref{apx:HamSim}.
\end{proof}

Now we present a lemma stating the complexity of implementing a block-encoding of an operator accessed via sparse matrix input oracles. As a corollary we get that the quantum operator input model generalizes the sparse-matrix input model, when we set $\alpha=s$.

\begin{lemma}\label{lem:sparseBlockEncoding}
	Suppose that $A$ is a $w$-qubit $s$-sparse operator given in the sparse matrix input model (Oracles \eqref{eq:oracleind} and \eqref{eq:oraclemat} with respect to a fixed $j$). Then we can implement an $(s,\BOT{\frac{s}{\eps}}{w},\eps)$-block encoding of $A$ with $\bigO{1}$-queries, and $\BOT{\frac{s}{\eps}}{w}$-other two-qubit quantum gates.
\end{lemma}
\begin{proof}
	See e.g., \cite{BerryChilds:hamsimFOCS}. \anote{Cite paper with Nathan.}
\end{proof}

We present a lemma based on ideas of~\cite[Corollary 9]{LowChuang:Qubitization2016} showing how to implement a block-encoding of a (subnormalized) density operator using purified access to the density operator. 

\begin{lemma}[Block-encoding of a (subnormalized) density operator \cite{LowChuang:Qubitization2016}]\label{lem:densitySimulation}
	Let $G$ be a $(w+a)$-qubit unitary which on the input state $\ket{0}^{\!\otimes w}\ket{0}^{\!\otimes a}$ prepares a purification $\ket{\srho}$ of the subnormalized $w$-qubit density operator $\srho$.
	Then we can implement a $(1,w+a,0)$-block-encoding of $\srho$, with single use of $G$ and its inverse and with $w+1$ two-qubit gates.
\end{lemma}
\begin{proof}
	Let us write $\ket{\srho}=\alpha\ket{\rho_0} +\beta\ket{\rho_1}$, where $\alpha,\beta\in\mathbb{R}$, $(I_{2^w}\otimes\ketbra{0}{0}\otimes I_{2^{a-1}})\ket{\rho_0}=\ket{\rho_0}$ and $(I_{2^w}\otimes\ketbra{1}{1}\otimes I_{2^{a-1}})\ket{\rho_1}=\ket{\rho_1}$. 
	Moreover, without loss of generality we can assume that $\ket{\srho_0}=\sum_{j=1}^{D}\sqrt{p_j}\ket{\psi_j}\ket{0}\ket{\tilde{\psi}_j}$ such that $\braket{\tilde{\psi}_i}{\tilde{\psi}_j}=\delta_{ij}$ and $\srho=\alpha^2\sum_{j=1}^Dp_j\ketbra{\psi_j}{\psi_j}$.
	Consider the $(2w+a+1)$-qubit unitary $V=(I_{2^{w+1}}\otimes G^\dagger) (\mathrm{SWAP}_{w+1} \otimes I_{2^{a}}) (I_{2^{w+1}}\otimes G)$,
	where $\mathrm{SWAP}_{w+1}$ denotes the unitary which swaps the first $w+1$-qubit register with the second $w+1$-qubit register. 
	Observe that 
	\begin{align*}
	(I_{2^w}\otimes \bra{0}^{\otimes 1+w+a})V(I_{2^w}\otimes \ket{0}^{\otimes 1+w+a})
	&=\left(I_{2^w}\otimes \bra{0}\bra{\srho}\right)(\mathrm{SWAP}_{w+1} \otimes I_{2^a})\left(I_{2^w}\otimes \ket{0}\ket{\srho}\right)\\	
	&=\alpha^2\left(I_{2^w}\otimes \bra{0}\bra{\rho_0}\right)(\mathrm{SWAP}_{w+1} \otimes I_{2^a})\left(I_{2^w}\otimes \ket{0}\ket{\rho_0}\right)\\		
	&=\alpha^2\sum_{j=1}^{D}p_j\left(I_{2^w}\otimes \bra{0}\bra{\psi_j}\bra{0}\right)(\mathrm{SWAP}_{w+1} \otimes I_{2^a}) \left(I_{2^w}\otimes \ket{0}\ket{\psi_j}\ket{0}\right)\\
	&=\alpha^2\sum_{j=1}^{D}p_j\left(I_{2^w} \otimes \bra{0}\otimes\bra{\psi_j}\bra{0}\right) \left(\ket{\psi_j}\ket{0}\otimes I_{2^w}\otimes\ket{0}\right)\\
	&=\alpha^2\sum_{j=1}^{D}p_j\left(\ket{\psi_j} \otimes \bra{\psi_j}\right)\\
	&=\alpha^2\sum_{j=1}^{D}p_j \ketbra{\psi_j} {\psi_j}\\
	&=\srho.		
	\end{align*}
	\vskip-6mm
\end{proof}
The above corollary essentially shows that the quantum operator input model generalizes the quantum state input model too, by choosing $\alpha=B$ (since $\mu^+_j - \mu^-_j + |\mu^I_j|\leq B$). What is left is to show how to implement a linear combination of block-encodings, e.g., $A_j = \mu^+_j \srho^+_j - \mu^-_j \srho^-_j + \mu^I_j I$.
We show how to efficiently implement such a block-encoding in the next subsection.

\subsubsection{Implementing a linear combination of block-encodings}
Following the work of Gilyén and Wiebe~\cite{gilyenBlockMatrices}, in this subsection we show how to efficiently implement a linear combination of block-encodings. Together with the optimal block-Hamiltonian simulation theorem from the previous subsection this enables a clean and efficient way to implement Gibbs-sampling using the techniques developed in~\cite[Appendix B]{AGGW:SDP}. Using linear combinations of block-encodings by-passes the entrywise summation of the input matrices which was a major bottleneck in previous SDP-solvers for the sparse input model~\cite{AGGW:SDP}. 

We define state-preparation unitaries in order to conveniently state our next lemma about implementing a linear combinations of block-encodings.

\begin{definition}[State-preparation pair]
	Let $y\in \mathbb{C}^m$ and $\beta\geq \nrm{y}_1$, the pair of unitaries $(P_L,P_R)$ is called a $(\beta, b, \eps)$-state-preparation pair for $y$ if $P_L\ket{0}^{\otimes b}= \sum_{j=0}^{2^b-1} c_j \ket{j}$ and  $P_R\ket{0}^{\otimes b}= \sum_{j=0}^{2^b-1} d_j \ket{j}$ such that $\sum_{j=0}^{m-1}|\beta \cdot (c^*_j d_j) -y_j | \leq \eps$ and for all $j\in m,\ldots, 2^b-1$ we have $c^*_j d_j =0$. A \emph{symmetric} state-preparation pair also satisfies $c_j=d_j$ for all $j\in 0\ldots m-1$.
\end{definition}

\begin{lemma}[Linear combination of block encodings]\label{lem:linCombBlocks}
	Let $A=\sum_{j=0}^{m}y_j A_j$ be a $w$-qubit operator and $\eps\in \mathbb{R}_+$. Suppose that $(P_L,P_R)$ is a $(\beta, b, \eps_1)$-state-preparation pair for $y$, $W=\sum_{j=0}^{m-1} U_j \otimes \ketbra{j}{j}+(I_{2^{w+a}}\otimes(I_{2^{b}}-\sum_{j=0}^{m-1}\ketbra{j}{j}))$ is an $w+a+b$ qubit unitary such that for all $j\in 0,\ldots, m$ we have that $U_j$ is an $(1,a,\eps_2)$-block-encoding of $A_j$. Then we can implement a $(\beta,a+b,\eps_1+\beta\eps_2)$-block-encoding of $A$, with a single use of $W$, $P_R$ and $P_L^\dagger$.
\end{lemma}
\begin{proof}
	Observe that $$\widetilde{W}=(I_{2^w}\otimes I_{2^a} \otimes P_L^\dagger) W  (I_{2^w}\otimes I_{2^a} \otimes P_R)$$ is 
	$(\beta,a+b,\eps_1+\beta\eps_2)$-block-encoding of $A$.
\end{proof}




Now we describe how to use a quantum-access classical RAM (QCRAM) to efficiently implement a state-preparation-pair unitary that can be used to construct the linear combinations of the block-encodings. In the SDP-solver we use this data structure for the summation of constraint matrices needed for Gibbs-sampling.
\begin{lemma}\label{lem:data}
	There is a data structure that can store an $m$-dimensional $d$-sparse vector $y$ with $\theta$-precision using a QCRAM of size $\BOT{\frac{m}{\theta}}{d}$. Furthermore:
	\begin{itemize}
		\item Given a classical $\bigO{1}$-sparse vector, adding\footnote{In order to avoid error accumulation from repeated roundings, we assume for simplicity that there can be at most $\bigO{\poly(m/\theta)}$ such calls to the data structure in total.} it to the stored vector has classical cost $\BOT{\frac{m}{\theta}}{1}$. 
		\item Given that $\beta\geq \nrm{y}_1$ we can implement a (symmetric) $(\beta, \BOT{\frac{m}{\theta}}{1}, \theta)$-state-preparation pair for $y$ with $\BOT{\frac{m}{\theta}}{1}$ queries to the QCRAM.
	\end{itemize}
\end{lemma}
\begin{proof}
	We use the data structure of \cite[Appendix A]{KerenidisQuantumRecommendation}.
\end{proof}

\begin{corollary}\label{cor:GibbsNormal}
	Having access to the above data structure for $y$, we have $T_{\mathrm{Gibbs}}(K,d,\theta)=\BOT{\theta}{\alpha K\sqrt{n}}$ in the quantum operator input model.
\end{corollary}
\begin{proof}
	 This can be proven using the Gibbs-sampler of \cite{AGGW:SDP}, combined with the above lemma for simulating the linear combination of the operators $\sum_{j=0}^m y_jA_j$.
\end{proof}
This directly gives the following result for SDP-solving:
\begin{theorem}\label{thm:simpGibbs}
In the quantum operator input model
\[
T_{SDP}(\eps) =\bOt{(\sqrt{m} + \sqrt{n}\gamma)\alpha\gamma^4},
\]
where $\gamma = Rr/\eps$ is. For a primal oracle the same complexity can be accomplished with $\gamma = R/\eps$. The  input oracle of the quantum operator model can be constructed using $\bigO{1}$ queries and $\BOT{mn\alpha\gamma}{1}$ elementary operations in the sparse matrix model and also in the quantum state model with setting $\alpha = s$ or $\alpha = B$ respectively. Therefore the above bound applies to these input models too.
\end{theorem}
\begin{proof}
	The complexity statement follows from Theorem~\ref{thm:mainSDP} using Corollary~\ref{cor:specTraceCalc} and \ref{cor:GibbsNormal}. The reductions follow from Lemma~\ref{lem:sparseBlockEncoding}-\ref{lem:densitySimulation}.
\end{proof}
If we do not apply Two-Phase Minimum Finding but use standard quantum minimum finding~\cite{durr&hoyer:minimum} instead, then we get a result with a slightly better dependence on $\gamma$, cf.~\eqref{eq:SDPCostAmpEst}: 
\begin{corollary}\label{cor:simpGibbsNoOR}
In the quantum operator input model
\[
T_{SDP}(\eps) =\bOt{\sqrt{mn}\alpha\gamma^4},
\]
where $\gamma = Rr/\eps$ for a full SDP-solver and $\gamma = R/\eps$ for a SDP primal oracle. 
\end{corollary}

\subsection{Exponentially improved Gibbs-sampling in the quantum state input model}\label{subsec:stategibbs}
In this subsection we show how to harness the special structure of the quantum state input model, to improve the complexity of Corollary~\ref{cor:GibbsNormal}. As shown in~\cite{brandao:expSDP2} this allows for an SDP-solver with a polylog dependence on $n$. We improve on the results of~\cite{brandao:expSDP2} by constructing a Gibbs-sampler with no explicit dependence on the rank of the input matrices. Moreover, we also improve the dependence on precision from polynomial to logarithmic.

We will use the following lemma about projectors.
\begin{lemma}\label{lem:projectors}
Let $0 < q< 1$, $\Pi$ be a projector and $\srho$ a subnormalized density operator.
Suppose that $q\Pi\preceq \srho$, $(I-\Pi)\srho(I-\Pi)=0$ and we have access to an $a$-qubit unitary $U_{\asrho}$ preparing a purification of a subnormalized density operator $\asrho$ such that $\nrm{\srho - \asrho}_1\leq 4\nu$. Then we can a prepare a purification of a subnormalized density operator $\asrho_{\mathrm{unif.}}$ such that $\nrm{\frac{q}{4}\Pi - \asrho_{\mathrm{unif.}}}_1\leq \bOt{\nu/q}$,	with $\BOT{\nu}{1/q}$ queries to $U_{\asrho}$ and its inverse and using $\BOT{\nu}{a/q}$ two-qubit gates.
\end{lemma}
\begin{proof}
	First let us assume that we have access to $U_{\srho}$ instead of $U_{\asrho}$. Then, using Corollary~\ref{cor:NegativePowerCost} from Appendix~\ref{apx:HamSim}, we could implement a unitary $W$ which is a $(1,\BOT{\nu,q}{a},0)$-block-encoding of $V$ such that $\nrm{\left(V-\frac{\sqrt{q}}{2\sqrt{\srho}}\right)\Pi}\leq\nu$. Note that since $q\Pi\preceq \srho$ and $(I-\Pi)\srho(I-\Pi)=0$ we know that $\srho$ is supported on the image of $\Pi$, in particular $\Pi \srho \Pi=\srho$. Using Hölder's inequality it is easy to see that
	\[
	\nrm{V\srho V^{\dagger} - \frac{\sqrt{q}}{2\sqrt{\srho}}\srho \frac{\sqrt{q}}{2\sqrt{\srho}}}_{1} \leq 2\nu,
	\]
	which is equivalent to saying
	\[
	\nrm{V\srho V^{\dagger} - \frac{q}{4}\Pi}_{1} \leq 2\nu.
	\]
	Corollary~\ref{cor:NegativePowerCost} shows that $W$ can be implemented with a single use of a controlled Hamiltonian simulation unitary $e^{it\srho}$, with maximal simulation time $|t|\leq \BOT{\nu}{1/q}$.
	This translates to $\BOT{\nu}{1/q}$ uses of $U_{\srho}$ and its inverse as shown by Lemma~\ref{lemma:controlledHamsin}. 
	
	Considering that $\nrm{\srho - \asrho}_\infty\leq \nrm{\srho - \asrho}_1\leq 4\nu$, if in the implementation of the controlled Hamiltonian simulation we replace $U_{\srho}$ by $U_{\asrho}$, then we make no bigger error than $\bOt{\nu/q}$, as shown by Theorem~\ref{thm:blockHamSim}.
	The resulting new unitary $\tilde{W}$ will be therefore an $(1,\BOT{\nu,q}{a},\bOt{\nu/q})$-block-encoding of $V$. Therefore we can prepare a purification of the a subnormalized density operator $\asrho_{\mathrm{unif.}}$ such that 
	\[
	  \nrm{\frac{q}{4}\Pi - \asrho_{\mathrm{unif.}}}_1\leq \bOt{\nu/q}.
	\]
\end{proof}	
\noindent In the proof of the next lemma we will mostly be looking at \emph{Eigenvalue threshold projectors}.
\begin{definition}[Eigenvalue threshold projector]
Suppose $H$ is a Hermitian matrix and $q \in \mathbb{R}$. Let $\Pi_{H > q}$ denote the orthogonal projector corresponding to the subspace spanned by the eigenvectors of $H$ that have eigenvalue larger than $q$. We define $\Pi_{H \leq q}=I-\Pi_{H > q}$ in a similar way.
\end{definition}

We are now ready to prove the main lemma of this section, an improved Gibbs-sampler in the quantum state model. In the proof we will use some specific conventions and notation to simplify the form of the proof. We say that two subnormalized density operator are $\delta$-close when their trace distance is at most $\delta$, and that two unitaries are $\delta$-close when their operator norm distance is at most $\delta$. We will always work with purifications of subnormalized density operators, so when for example we say that we apply an operator to a subnormalized density operator, we mean that we apply the operator to its purification. 

\begin{lemma}[Gibbs-sampling of the difference of density operators]\label{lemma:preGibbs}
	Suppose that we have unitaries $U_{\srho^{\pm}}$ preparing a purification of the subnormalized density operators $\srho^\pm$ using $a=\bigO{\mathrm{poly}\log(n)}$ qubits.
	Let\footnote{In case $n$ is not a power of $2$ we still represent $H$ on $\left\lceil\log_2(n)\right\rceil$ qubits, but think about it as an $\mathbb{C}^{n\times n}$ operator.} $H:=(\srho^+-\srho^-)/2$, $\beta \in[1,n/2]$ and $\delta,\eta\in(0,1]$. Assume we are given a point $q\in [2/n,1/\beta]$ that is $\eta$-far from the spectrum of $H$, i.e.
\[
 |\lambda-q|  \geq \eta \text{ for all }\lambda \in \mathrm{Spec}(H).
\]
	Then we can prepare a purification of an approximate Gibbs-state $\tilde{\rho}_{Gibbs}$ such that 
\[
\nrm{\frac{e^{\beta H}}{\tr{e^{\beta H}}}-\tilde{\rho}_{Gibbs}}_1\leq \delta
\]
with\footnote{We think that it should be possible to improve the complexity to $\BOT{\frac{n}{\delta}}{q^{-1}/\eta}$ using recent results about variable-time amplitude amplification and estimation~\cite{CGJ:PowerOfBlockPowers18}.} $\BOT{\frac{n}{\delta}}{q^{-1.5}/\eta}$ queries to controlled-$U_{\srho^{\pm}}$ and their inverses.\footnote{If $U_{\srho^{\pm}}$ are not controlled, then it is easy to construct a controlled version using $\bigO{a}$ extra ancilla qubits.}
\end{lemma}
\begin{proof}
The main idea of the proof is that we prepare (slightly subnormalized) density operators corresponding to $\Pi_{H>q}$ and $\Pi_{H\leq q}$, i.e., uniform distributions over a partition of eigenspaces of $H$. Utilising these states we prepare subnormalized Gibbs states on the corresponding subspaces, then merge and amplify the states in order to obtain the final Gibbs state. This is beneficial since on the subspace corresponding to $\Pi_{H\leq q}$ the map $e^{\beta H}$ is nicely bounded. However, on the image of $\Pi_{H>q}$ the map $e^{\beta H}$ might behave wildly, and in the extreme case this map might magnify the amplitude of some eigenvectors tremendously. This implies that we need to ``find'' such magnified elements, as the Gibbs state is concentrated around them. Fortunately the rank of $\Pi_{H>q}$ is at most $1/q$, which makes it easier to ``find'' the extreme vectors then if we would apply the same procedure to the uniform distribution $I/n$.

We start with implementing the unitary $\tilde{V}_{H>q}$, that labels eigenstates of $H$ corresponding to which component of $\mathbb{R}\setminus\{q\}$ their eigenvalue lies in. More precisely, we set $\delta':= \tilde{\Theta}(\delta q^2)$, and we want to implement a unitary $\tilde{V}_{H>q}$ which is a 
$(1,\BOT{\frac{n}{\eta\delta'}}{1},\delta')$-block-encoding 
of $(\Pi_{H>q}\otimes I + \Pi_{H\leq q}\otimes X)$. 
Due to the assumption that $q$ lies at least $\eta$-far from $\mathrm{Spec}(H)$ we can implement these unitaries using $\Theta(\eta)$ precise phase estimation of the operator $e^{iH}$, repeated $\bigO{\log(1/\delta')}$ times. This can be implemented with $\BOT{\delta'}{1/\eta}$ queries as show by Lemmas~\ref{lem:densitySimulation} and \ref{lem:linCombBlocks}.  
\\\\
\noindent Now let us consider Gibbs-sampling on the image of $\Pi_{H > q}$. Let $\ssigma^+:= (\srho^+ + \srho^-)/2$ be a subnormalized density operator, which we can prepare in a purified manner using $\bigO{1}$ queries. 
Also let 
\[
\ssigma^+_{H>q}:=  \Pi_{H>q}\ssigma^+\Pi_{H>q},
\]
and observe that we can prepare $\assigma^+_{H>q}$, such that 
\begin{equation}\label{eq:sigmaApx}
   \nrm{\assigma^+_{H>q}-\ssigma^+_{H>q}}_1\leq 2\delta'(\leq q/4),
\end{equation}
by applying $\tilde{V}_{H>q}$ to $\ssigma^+$ (the second inequality can be assumed w.l.o.g. since $\delta' = \tilde{\Theta}(\delta q^2)$).

Observe that
\begin{align}
	q \Pi_{H>q} &= \Pi_{H>q}(q \Pi_{H>q})\Pi_{H>q}\nonumber\\
	&\preceq \Pi_{H>q}(H)\Pi_{H>q}\nonumber\\
	&\preceq \Pi_{H>q}(H+2\srho^-)\Pi_{H>q}\nonumber\\
	&=\ssigma^+_{H>q}.\label{eq:claim1}
\end{align}
This allows us to apply Lemma~\ref{lem:projectors} to $\srho:=\ssigma^+_{H>q}$, $\asrho:=\assigma^+_{H>q}$ and $\Pi:=\Pi_{H>q}$ with $\nu:= \delta'$ so we get that we can prepare a state $\asrho_{\mathrm{unif.}}$ such that 
\[
\nrm{\frac{q}{4}\Pi_{H>q} - \asrho_{\mathrm{unif.}}}_1\leq \bOt{\delta'/q}
\]
 using $\BOT{\delta'}{q^{-1}/\eta}$ queries.
	
Now we can check if $\Pi_{H>q}=0$ or not as follows. If it is not $0$ then $\tr{\Pi_{H>q}}\geq 1$ and hence by \eqref{eq:sigmaApx}-\eqref{eq:claim1} we have $\tr{\assigma^+_{H>q}}\geq\tr{\ssigma^+_{H>q}} -q/4 \geq q-q/4=3q/4$. Since $\ssigma^+_{H>q}=\tr{\Pi_{H>q}}\ssigma^+_{H>q}\tr{\Pi_{H>q}}$, when $\Pi_{H>q}=0$ it similarly follows that $\tr{\assigma^+_{H>q}}\leq \tr{\ssigma^+_{H>q}}+ q/4=q/4$. Thus we can check whether $\Pi_{H>q}=0$ by checking whether $\tr{\assigma^+_{H>q}}\leq q/4$ or $\tr{\assigma^+_{H>q}}\geq 3q/4$. This can be done with success probability at least $1-\delta'$ by using amplitude estimation with $\BOT{\delta'}{q^{-0.5}}$ calls to the procedure preparing $\assigma^+_{H>q}$, costing $\BOT{\delta'}{q^{-0.5}/\eta}$ queries in total. For the final Gibbs-sampling we will consider the Gibbs state on the image of $\Pi_{H>q}$ and $\Pi_{H\leq q}$ separately. Therefore if $\Pi_{H>q}=0$ we only need to consider the Gibbs state on the image of $\Pi_{H\leq q}$, which we do later in this proof. For now we assume $\Pi_{H>q}\neq0$ and consider the Gibbs state on its image.

Now we use binary search in order to find $\lambda_{\max}$ the maximal eigenvalue of $H$, with precision $\beta^{-1}/2$ and success probability $1-\delta'$. By our assumption $\lambda_{\max}\in(q,1]$. We start each iteration of the binary search by performing phase estimation on $\assigma^+_{H>q}$ using the unitary $e^{iH}$ with precision $\beta^{-1}/4$ and success probability $1-q/4$. By \eqref{eq:sigmaApx}-\eqref{eq:claim1} we know that the eigenvector corresponding to $\lambda_{\max}$ is present with probability at least $3q/4$ in $\assigma^+_{H>q}$, and the other eigenvalues are present with a probability at most $1$ in total. Therefore the probability of obtaining a phase estimate $\tilde{\lambda}$ such that $\tilde{\lambda}\geq \lambda_{\max} - \beta^{-1}/4$ is at least $q/2$, whereas the total probability of obtaining a phase estimate $\lambda'$ such that $\lambda'\geq \lambda_{\max} + \beta^{-1}/4$ is at most $q/4$. Therefore we can perform each iteration of the binary search with precision $\beta^{-1}/2$ and success probability $1-\delta'/\log(q^{-1})$ by applying amplitude estimation to the probability of getting an eigenvalue estimate from the current search interval, using $\BOT{\beta}{q^{-0.5}}$ repetitions of the initial state preparation and phase estimation procedure. Thus each iteration has query complexity $\BOT{\beta}{q^{-0.5}(\beta+1/\eta)}=\bOt{q^{-1.5}+q^{-0.5}/\eta}$, giving the same total query complexity bound $\bOt{q^{-1.5}+q^{-0.5}/\eta}$ for the complete binary search.


After finding the minimum up to precision $1/(2\beta)$ we can compute a number $\tilde{\lambda}_{\max}$ such that 
\begin{equation}\label{eq:minEigenValueGuarantee}
	\tilde{\lambda}_{\max}I\succeq H \text{  but  }(\tilde{\lambda}_{\max}-1/\beta)I\nsucc H.
\end{equation}
Using Lemmas~\ref{lem:densitySimulation}-\ref{lem:linCombBlocks}, and the results of \cite[Appendix B]{AGGW:SDP} we can implement an $(1/2,\BOT{\frac{\beta}{\delta'}}{1},\delta')$-block encoding of $e^{\beta\frac{H-\tilde{\lambda}_{\max}I}{2}}$ using $\BOT{\delta'}{\beta}$ queries. Applying this map to $\asrho_{\mathrm{unif.}}$ gives an $\bOt{\delta'/q}$-approximation of the subnormalized density operator $\frac{qe^{-\beta\tilde{\lambda}_{\max}}}{16}\Pi_{H>q} e^{\beta H}$. Observe that since we assumed $\tr{\Pi_{H>q}}\geq 1$, by \eqref{eq:minEigenValueGuarantee} we get that
\begin{equation}\label{eq:subNormHighSpace}
\tr{\frac{qe^{-\beta\tilde{\lambda}_{\max}}}{16}\Pi_{H>q} e^{\beta H}}\geq q/(16e).
\end{equation}
 Thus we can prepare a subnormalized $\bOt{\delta'/q}$-approximation of the Gibbs state on the image of $\Pi_{H>q}$ having trace at least $q / (16e)$.
\\\\
\noindent Now we consider the Gibbs state on $\Pi_{H\leq q}$. 
First observe that we can prepare the density operator $I/n$ by preparing a the maximal entangled state $\frac{1}{\sqrt{n}}\sum_{j=1}^n \ket{j}\ket{j}$ using $\bOt{\log(n)}$ two-qubit quantum gates.
With a single use of the unitary $\tilde{V}_{H>q}$ we can prepare an $\bigO{\delta'}$ approximation of $\frac{1}{n}\Pi_{H\leq q}$ by simply marking the appropriate eigenstates of $I/n$, which takes $\bOt{1/\eta}$ queries. Then we apply the map $e^{\beta\frac{H}{2}}/(2\sqrt{e})$ on the subspace $\Pi_{H\leq q}$ with $\delta'$ accuracy. Since $q\leq 1/\beta$ all eigenvalues of $\beta H/2$ that we are concerned with are smaller in absolute value than $1/2$.  As shown in \cite[Appendix B]{AGGW:SDP} this implies that implementing the map $e^{\beta\frac{H}{2}}/(2\sqrt{e})$ with $\delta'$ precision requires Hamiltonian simulation of $\beta H$ for constant time, which can be done using $\BOT{\delta'}{\beta}$ queries. Therefore we can prepare a $\bigO{\delta'}$ approximation of the state $\frac{1}{4 e n}\Pi_{H\leq q}e^{\beta H}$ with $\BOT{\delta'}{\beta +1/\eta}\leq \BOT{\delta'}{q^{-1}/\eta}$ queries.

Like before, we would like to lower bound the trace of the created subnormalized density operator. First note that $\nrm{H}_1\leq 1$, and so the number\footnote{We count eigenvalues with algebraic multiplicity.} of eigenvalues that are larger than $q$ in absolute value is at most $1/q \leq n/2$, thus $\tr{\Pi_{|H|\leq q}}\geq n/2$. Also note that
$\Pi_{H\leq q}e^{\beta H}\succeq \Pi_{|H|\leq q}e^{\beta H}$, and for an eigenvalue $\lambda$ such that $|\lambda|\leq q$ we have $e^{\beta\lambda} \geq e^{-\beta q} \geq 1/e$. It follows that
\begin{equation}\label{eq:subNormLowSpace}
\tr{\frac{1}{4 e n}\Pi_{H\leq q}e^{\beta H}} \geq\tr{\frac{1}{4 e n}\Pi_{|H|\leq q}e^{\beta H}}\geq\tr{\frac{1}{4 e^2 n}\Pi_{|H|\leq q}} \geq \frac{1}{8e^2}.
\end{equation}
\\\\
\noindent As we can now Gibbs-sample on both parts of the spectrum, we are ready to combine the two. 
Let\footnote{In the special case when $\Pi_{H>q}=0$ we simply set $\xi:=\frac{1}{4 e n}$.} $\xi:=\min\left(\frac{qe^{-\beta \tilde{\lambda}_{\max}}}{16},\frac{1}{4 e n}\right)$, then we can prepare a purification of $\asrho_{G}$ which is an $\bOt{\delta'/q}$-approximation of
\[
\srho_{G}:=\frac{\xi}{2}e^{\beta H}= 
	\underset{\leq 1/2}{\underbrace{\left(\frac{\xi}{2}\frac{16}{qe^{-\beta \tilde{\lambda}_{\max}}}\right)}} \frac{qe^{-\beta \tilde{\lambda}_{\max}}}{16}\Pi_{H>q} e^{\beta H} 
	+ \underset{\leq 1/2}{\underbrace{\left(\frac{\xi}{2}\frac{4 e n}{1}\right)}}\frac{1}{4 e n}\Pi_{H\leq q}e^{\beta H},
        \]
	by mixing the two subnormalized Gibbs states on the corresponding subspaces with appropriate $(\leq 1/2)$ coefficients. This subnormalized Gibbs state $\asrho_{G}$, can be prepared at the same cost as the two partial Gibbs-state preparation, that is $\bOt{q^{-1}/\eta}$ queries\footnote{Note that we do the maximum finding to find $\tilde{\lambda}_{\max}$ only once, and we do not count its complexity in the state preparation.} as we have already shown.
	
	Note that $\tr{\srho_{G}}=\Omega(q)$ as shown by \eqref{eq:subNormHighSpace}-\eqref{eq:subNormLowSpace}, therefore we can use $\bOt{\sqrt{1/q}}$ amplitude amplification steps to prepare an $\bOt{\delta'/q^2}$ approximation of $\frac{\srho_{G}}{\tr{\srho_{G}}}$, which is clearly $\rho_{Gibbs}$. In total this yields an $\bOt{q^{-1.5}/\eta}$ query algorithm. Since $\delta' =  \tilde{\Theta}(\delta q^2)$ this concludes the proof.
\end{proof}

The following corollary expands our new Gibbs-sampling result, it gives an exponential improvement in terms of the precision over the previous approach for this input model by Brand\~ao et al.~\cite{brandao:expSDP2}. The dependence on the success probability is worse, but in our application to SDP-solving we only require success for a constant fraction of random seeds. Furthermore, there is no longer a dependence on the rank of the input matrices.

\begin{theorem}\label{thm:QGibbs}
	Suppose we have query access to the $a=\mathrm{poly}\log(n)$-qubit unitaries $U_{\srho^{\pm}}$ preparing a purification of the (subnormalized) density operators $\srho^\pm\in\mathbb{C}^{n\times n}$, such that $H=(\srho^+-\srho^-)/2$, $\beta\geq 1$, $\theta,\delta\in(0,1]$. 
	Then there is a quantum algorithm, that using\footnote{Similarly to Lemma~\ref{lemma:preGibbs} we think that it should be possible to improve the complexity to $\BOT{\theta}{\beta^{3}/\delta}$ using recent results about variable time amplitude amplification and estimation~\cite{CGJ:PowerOfBlockPowers18} (maybe at the expense of worse dependence on the error).} $\BOT{\theta}{\beta^{3.5}/\delta}$ queries to controlled-$U_{\srho^{\pm}}$ or their inverses, prepares a purification of a quantum state $\rho_{S}$ such that 
	$$ \nrm{\rho_{S}-\frac{e^{-\beta H}}{\tr{e^{-\beta H}}}}_1\leq \theta, $$
	where $S$ is an $\bigO{\log(\beta/\delta)}$-bit random seed, and the above holds for at least $(1-\delta)$-fraction of the seeds.
\end{theorem}
\begin{proof}
	If $\beta\geq n/2$, then we simply use the Gibbs-sampler from Theorem~\ref{thm:simpGibbs}.
	Otherwise, using the random seed $S$ we generate a uniform random number $q_S$ from the interval $[1/(2\beta),1/\beta]$. Note that since $\tr{|H|}\leq 1$ we have that $\left|\mathrm{Spec}(|H|)\cap [1/(2\beta),1/\beta]\right|\leq 2\beta$. Also the length of the interval is $1/(2\beta)$ therefore a random point in the interval falls $\delta/(8\beta^2)$-close to $\mathrm{Spec}(|H|)$ with probability at most $\delta$. Therefore the random seed can be used in such a way that $q_{S}$ will be $\eta=\delta/(8\beta^2)$ far from any point of $\mathrm{Spec}(|H|)$ with probability at least $1-\delta$. If this is the case the procedure of Lemma~\ref{lemma:preGibbs} prepares the sought Gibbs state with the stated complexity.
\end{proof}

\begin{corollary}\label{cor:GibbsExtra}
	Having access to the data structure of Lemma~\ref{lem:data} storing the vectors $\nu^{\pm}\in\mathbb{R}^{m+1}$ such that $\nu_j^{\pm} = y_j\mu^{\pm}_j$ for all $j\in 0\ldots m$, we have that $T_{\mathrm{Gibbs}}(K,d,\theta)=\BOT{\theta}{(BK)^{3.5}}$ using the quantum state input model.
\end{corollary}
\begin{proof}
To start, let us define
	$$H:=\sum_{j=0}^m \frac{y_jA_j}{KB}=\sum_{j=0}^m \frac{y_j}{KB}\left(\mu_j^+\srho_j^+-\mu_j^-\srho_j^-+\mu^I I\right)
	=\underset{\srho^+:=}{\underbrace{\sum_{j=0}^m \frac{y_j\mu_j^+}{KB}\srho_j^+}} 
	-\underset{\srho^-:=}{\underbrace{\sum_{j=0}^m \frac{y_j\mu_j^-}{KB}\srho_j^-}}
	+I\sum_{j=0}^m \frac{y_j\mu_j^I}{KB}.$$
Notice that we can ignore the identity terms since adding identities in the exponent does not change a Gibbs state,
also let
\[
\beta := KB \geq \sum_{j=0}^m y_j (\mu_j^++\mu^-_j).
\]
Using 
Lemma~\ref{lem:data} we can see that a $\bigO{\theta/\beta}$ approximation of  $\srho^{\pm}$ can be prepared with $\bigO{1}$ queries and using $\bigO{\mathrm{poly}\log(m\beta/\theta)}$ elementary operation.
By setting $\delta:=1/5$ and using Theorem~\ref{thm:QGibbs} the statement follows.
\end{proof}

This directly gives the following result for SDP-solving in the quantum state model
\begin{theorem}\label{thm:difGibbs}
	In the quantum state input model 
	\[
		T_{SDP}(\eps) = \bOt{  \left(\sqrt{m}+B^{2.5}\gamma^{3.5}\right)B\gamma^4},
	\]
	where $\gamma = Rr/\eps$. The same bound holds for a primal oracle with $\gamma = R/\eps$.
\end{theorem}
\begin{proof}
	This follow directly from Theorem~\ref{thm:mainSDP} using Corollaries \ref{cor:specTraceCalc},\ref{cor:GibbsNormal} and \ref{cor:GibbsExtra}.
\end{proof}
Also we can simply not apply Two-Phase minimum finding and use standard quantum minimum finding~\cite{durr&hoyer:minimum} instead to get a slightly better dependence on $\gamma$, cf.~\eqref{eq:SDPCostAmpEst}:
\begin{corollary}\label{cor:difGibbsNoOR}
	In the quantum state input model 
	\[
		T_{SDP}(\eps) = \bOt{\sqrt{m}B^{3.5}\gamma^{6.5}},
	\]
	where $\gamma = Rr/\eps$.. The same bound holds for a primal oracle with $\gamma = R/\eps$.
\end{corollary}

\section{Applications}\label{sec:applications}
In previous works on quantum SDP-solving~\cite{AGGW:SDP,brandao:quantumsdp} it remained an open question whether any applications could be found in the regime where $Rr/\eps$ was small enough to get a speedup over the best classical methods. Later Brand\~ao et al.~\cite{brandao:expSDP2} showed that SDP primal oracles can be used to solve the problem of shadow tomography if the input is given in the quantum state model. Shadow tomography was recently proposed by Aaronson~\cite{aaronson:shadow}, who also gave a sample-efficient algorithm. In Section~\ref{sec:shadow} we apply our improved SDP primal oracles to this problem, simultaneously improving the sample complexity and computational complexity compared to the previous works. 

We also propose new applications to quantum SDP-solvers, namely the problems of \emph{quantum state discrimination} and \emph{E-optimal design}. In both cases we show a speedup over the best possible classical algorithm in terms of some input parameters, while suffering from a massive dependence on other parameters. 

\subsection{Improved shadow tomography}\label{sec:shadow}
We apply the idea from Brand\~ao et al.~\cite{brandao:expSDP2} to use an SDP primal oracle to the problem of \emph{shadow tomography} proposed by Aaronson~\cite{aaronson:shadow}.
In \emph{shadow tomography} we are given the ability to sample from an $n$-dimensional quantum state $\tau$ and we have a description of some measurement operators $E_1,\dots, E_m$; the goal is to find $\eps$-approximations of the corresponding expectation values $\tr{E_j \tau}$ for all $j\in [m]$. Aaronson showed that this can be done with only
\[
\bOt{\frac{\log^4(m)\log(n)}{\eps^5}}
\]
samples from $\tau$, but his method has high computational costs.

In \cite{brandao:expSDP2} Brand\~ao et al.~showed that a slightly relaxed problem can be efficiently solved using an SDP primal oracle. The problem they solved is to find a $y\in\mathbb{R}^m$ for which $\sigma:=e^{- \sum_{j=1}^m y_j E_j} / \tr{ e^{- \sum_{j=1}^m y_j E_j}}$ is such that $|\tr{E_j(\tau-\sigma)}|\leq \eps/2\,\, \forall j\in [m]$, i.e., $\sigma$ is in
\begin{align*}
\mathcal{P}_{\eps} = \{ \sigma\colon &\sigma \succeq 0\\
  & \tr{\sigma} = 1\\
  & \tr{\sigma E_j} \leq \tr{\tau E_j } + \eps /2 \ \ \ \forall j\in [m]\\
  & \tr{- \sigma E_j} \leq \tr{ - \tau E_j } - \eps /2 \ \ \ \forall j\in [m]
\}.
\end{align*}
We call the problem of finding a classical description of $\tau$ that suffices to solve the shadow tomography problem without any more samples from $\tau$ the \emph{descriptive shadow tomography problem}.
In particular if we get a vector $y$ as above, then for a given $j\in[m]$ using $\BOT{m}{1/\eps^2}$ invocations of a Gibbs-sampler for $y$ followed by the measurement $E_j$ suffices to find an $\eps$-approximation of $\tr{\tau E_j}$ with success probability at least $1-\bigO{1/m}$. If we can coherently apply $E_j$, then using amplitude estimation techniques the number of (coherent) Gibbs-sampler calls can be reduced to $\BOT{m}{1/\eps}$.

Due to the output size of the shadow tomography problem, a trivial $\Omega(m\log(1/\eps))$ lower bound can be given on the computational complexity. However, this limitation does not exist for the descriptive shadow tomography problem.
Both problems clearly have the same sample complexity, furthermore the best known lower bound on the sample complexity is $\Omega(\log(m)/\eps^2)$~\cite{aaronson:shadow}.

\begin{theorem}
The descriptive shadow tomography problem can be solved using 
\[
\bOt{\frac{\log^4(m)\log(n)}{\eps^4}}
\]
samples from $\tau$.
Furthermore, when the $E_j$ matrices are accessible in the quantum operator model this can be done using
\[
\bOt{ \left(\sqrt{m} + \frac{\sqrt{n} }{ \eps}\right)\frac{\alpha}{\eps^4} } 
\]
queries. It follows that the same bound holds with $\alpha = s$ for the sparse model and with $\alpha = B$ for the quantum state model.
When the measurements are given in the quantum state model the query complexity can be also bounded by
\[
\BOT{n}{ \left(\sqrt{m} + \min\left(\frac{\sqrt{n}}{\eps},\frac{B^{2.5}}{\eps^{3.5}} \right) \right) \frac{B}{\eps^4} }.
\]
\end{theorem}
\begin{proof}
The samples from $\tau$ are only used for calculating the values $b_j$, i.e., $\eps/4$ approximations of $\tr{\tau E_j}$, when checking the constraints in the SDP primal oracle. Like in \cite{brandao:expSDP2} we make a small adjustment to our SDP primal oracle: when Gibbs-sampling the Gibbs state $\rho$, we also sample $\tau$ to create the state $\rho \otimes \tau$. Then, when checking the constraint, we measure $E_j \otimes -E_j$ to obtain an approximation of $\tr{E_j \rho} - \tr{E_j \tau}$. Notice that our SDP primal oracle uses $\bOt{\frac{\log^4(m)\log(n)}{\eps^4}}$ Gibbs states ($\BOT{\log(n)}{\log^4(m)/\eps^2}$ in each of the $\bigO{\log(n)/\eps^2}$ iterations) and hence the modified version uses that many samples from $\tau$ too.

The statement about the computational complexity follows directly from Theorem~\ref{thm:simpGibbs} and \ref{thm:difGibbs}.
\end{proof}

As a final remark, similarly to Low and Chuang~\cite{LowChuang:UniformSpectAmp2016}, we note that if one can perform a POVM measurement on a quantum computer, then one can also implement a block-encoding of the corresponding measurement operator. First we clarify what we mean by performing a POVM measurement on a quantum computer. For simplicity assume that the POVM is a two-outcome measurement, represented by the operators $M,(I-M)$. Then we assume (without too much loss of generality) that an implementation on a quantum computer is as follows:
We get as input a mixed state $\rho$, and attach $a$ ancilla qubits to it. Then we apply some unitary on the state, and finally perform a measurement in the computational basis, accepting only measurement outcomes where the last qubit is $\ket{0}$. We can summarise the procedure as follows:
 $$ \rho \rightarrow \rho\otimes\ketbra{0}{0}^{\otimes a}
 \rightarrow U \left(\rho\otimes\ketbra{0}{0}^{\otimes a}\right) U^\dagger  
 \rightarrow \tr{\left(I\otimes\ketbra{0}{0}\right) U \left(\rho\otimes\ketbra{0}{0}^{\otimes a}\right) U^\dagger}.$$
Suppose that the above implementation has at most $\eps$ bias, then 
\begin{equation}\label{eq:POVMeps}
\forall \rho\colon \left|\tr{ \left(I\otimes\ketbra{0}{0}\right) U \left(\rho\otimes\ketbra{0}{0}^{\otimes a}\right) U^\dagger}-\tr{M\rho}\right|\leq \eps.
\end{equation}
Observe that $\rho\otimes\ketbra{0}{0}^{\otimes a}= \left(I\otimes\ket{0}^{\otimes a}\right)\rho\left(I\otimes\bra{0}^{\otimes a}\right)$, and thus
$$ \tr{ \left(I\otimes\ketbra{0}{0}\right) U \left(\rho\otimes\ketbra{0}{0}^{\otimes a}\right) U^\dagger}
  =\tr{ \left(I\otimes\bra{0}^{\otimes a}\right)U^\dagger\left(I\otimes\ketbra{0}{0}\right) U  \left(I\otimes\ket{0}^{\otimes a}\right)\rho}.
$$
Therefore \eqref{eq:POVMeps} is equivalent to saying that
\begin{align*}
\forall \rho\colon \kern -18mm& &
\left|\tr{ \left[\left(I\otimes\bra{0}^{\otimes a}\right)U^\dagger\left(I\otimes\ketbra{0}{0}\right) U  \left(I\otimes\ket{0}^{\otimes a}\right)-M\right]\rho}\right|&\leq \eps \\
\Longleftrightarrow \kern -18mm& &
\nrm{\left(I\otimes\bra{0}^{\otimes a}\right)U^\dagger\left(I\otimes\ketbra{0}{0}\right) U  \left(I\otimes\ket{0}^{\otimes a}\right)-M}&\leq \eps.
\end{align*}
Finally let $a':=a+1$ and $U':=U\otimes I_2$, then it is easy to see that 
$$
	\left(I\otimes\bra{0}^{\otimes a}\right)U^\dagger\left(I\otimes\ketbra{0}{0}\right) U  \left(I\otimes\ket{0}^{\otimes a}\right)=	\left(I\otimes\bra{0}^{\otimes a'}\right)U'^\dagger\left(I\otimes\mathrm{CNOT}\right) U'  \left(I\otimes\ket{0}^{\otimes a'}\right),
$$
therefore $U'^\dagger\left(I\otimes\mathrm{CNOT}\right) U'$ is a $(1,a',\eps)$-block-encoding of $M$.

This shows that if we can implement the measurements $E_j$ in a controlled fashion on a quantum computer, then we can also implement the corresponding block encoding with essentially the same cost. Hence the descriptive shadow tomography problem can be solved with the same cost as $(\sqrt{m} +\sqrt{n}/\eps)/\eps^4$ controlled measurements of $E_j$, if the measurement is performed on a quantum computer as we described above.
\subsection{Quantum state discrimination}
In the \emph{Quantum State Discrimination} problem we are given $k$ $d$-dimensional quantum states $\rho^{(1)},\dots,\rho^{(k)}\in \mathbb{C}^{d\times d}$, in some oracular access model. Our goal is to find a POVM $M^{(1)},\dots,M^{(k)}$ that has the ``best" probability of discriminating between the states. Here ``best" can mean two things:
\begin{itemize}
	\item The minimal success probability is maximized: $\max_M \min_{i\in [k]}\tr{M^{(i)}\rho^{(i)}}$.
	\item The total success probability, the sum of all the success probabilities, is maximized.
\end{itemize}
Both problems can be cast as an SDP~\cite{eldar:sdp} but we will only consider the second here since it lends it self better to the Arora-Kale framework..
Our goal will be to get a quantum speedup in $k$ at the expense of a slowdown in terms of $d$. However, the interesting cases of the problem seem to occur when $d \ll k$. Furthermore, we will use the quantum state model not to get a further speedup over the sparse matrix access oracle, but to show that it is possible to solve the problem even when just given access to unitaries that prepare the quantum states. 


\begin{theorem}\label{thm:maxTotal}
	Given access to the matrix entries of the quantum states $\rho^{(1)},\dots,\rho^{(k)}\in \mathbb{C}^{d\times d}$ the total error quantum state discrimination problem can be solved up to additive error $\eps$ on a quantum computer using 
	\[
	\bOt{\frac{\sqrt{k}}{\eps^5} \poly(d)}
	\]
	queries to the input.
	
	Given access to a unitary that creates a purified version of the quantum states $\rho^{(1)},\dots,\rho^{(k)}\in \mathbb{C}^{d\times d}$ the total error quantum state discrimination problem can be solved up to additive error $\eps$ on a quantum computer using 
	\[
	\bOt{\frac{k^{1.5}}{\eps^5}\poly(d)}
	\]
	queries.
\end{theorem}
\begin{proof}
	To maximize the total success probability, notice that the probability of measuring $\rho^{(i)}$ correctly is $\tr{M^{(i)}\rho^{(i)}}$.
	Writing the problem as an SDP we get:
	\begin{align*} 
	\max \quad & \sum_{i=1}^k \tr{M^{(i)}\rho^{(i)}} \\ 
	\text{s.t.}\ \ \ & \sum_{i=1}^k M^{(i)} = I_{d}\\
	&M^{(i)} \succeq 0 \text{ for all }i\in [k].
	\end{align*}
	This can be written in the standard form~\eqref{eq:SDP} as follows:
	\begin{itemize}
		\item $X = \diag\left(M^{(1)},\dots,M^{(k)} \right)$.
		\item $C = \diag\left(\rho^{(1)},\dots,\rho^{(k)} \right)$.
		\item $A_{st} = \oplus_{i=1}^k E_{st}$ and $b_{st} = \delta_{st}$, for $\delta_{st}$ the Kronecker delta.
	\end{itemize}
	
	Notice that we have strict equalities, as opposed to the inequalities in the standard form. These equalities can be cast into inequality form by adding a separate upper and lower bound, this is however not needed for the analysis. Just note that an equality in the primal corresponds to a variable in the dual without positivity constraint.
	
	To apply our SDP-solvers we need to give bounds on the input parameters. Clearly here $n = kd$ and $m=\bigO{d^2}$. Furthermore, since the objective matrix is block diagonal with $d\times d$ blocks, the sparsity $s$ is at most $d$. To bound $B$, note that $C$ has trace $k$ and is psd, and all other constraints can clearly be decomposed with a constant trace. It remains to give a bound for $R$ and $r$.
	
	For $R$, the bound on the trace of a primal solution, notice that
	\[
	\tr{X} = \sum_{i=1}^k \tr{M^{(i)}} = \tr{\sum_{i=1}^k M^{(i)}} = \tr{I_d} = d.
	\]
	For $r$ we need to write out the dual, doing so directly gives:
	\begin{align*}
	\min \quad &  \sum_{s,t}^d y_{st} \delta_{st}  = \sum_{s}^d y_{ss}\\ 
	\text{s.t.}\ \ \ &\sum_{j=1}^m y_{st} \left( \oplus_{i=1}^k E_{st} \right) \succeq \oplus_{i=1}^k \rho^{(i)}.
	\end{align*}
	Notice that we do not have a $y\geq 0$ since we have strict equalities in the primal. We could have replaced the equalities by inequalities and then we would get a $y^+$ and $y^-$ vector, both non-negative, such that $y = y^+-y^-$. However, since $r$ is a bound on the sum of the values in \emph{one} optimal solution, it is enough to bound the absolute value of the $y$ variables.
	
	To do so, simply rewrite the dual in block form and reorganize the $y$ variables in a matrix $Y$:
	\begin{align*}
	\min \quad &  \tr{Y}\\ 
	\text{s.t.}\ \ \ & Y \succeq  \rho^{(i)} \text{ for all }i\in[k].
	\end{align*}
	Clearly $I_d$ is feasible for this problem so for an optimal $Y$ we have $\opt = \tr{Y} \leq d$. This gives the bound, for $S\in\{-1,1\}^{d\times d}$
	\begin{equation}\label{eq:rOPTBound}
	\sum_{s,t}^d | Y_{st} | = \max_{S\in\{-1,1\}^{d\times d}}\tr{S Y} \leq  \tr{Y} \max_{S\in\{-1,1\}^{d\times d}}\nrm{S} \leq d\tr{Y} \leq d^2,
	\end{equation}
	so $r = d^2$ suffices.\footnote{This also proves that for $k$ states of dimension $d$ the total success probability of discrimination is always at most $d$, so the average will be at most $d/k$. Thus the error parameter should scale with $1/k$ if we would consider the average probability. This is why we choose to look at the total success probability instead.}
	
	Applying our results about the complexity of SDP-solving gives the claimed bounds.
\end{proof}

Note that the output of the algorithm is a classical description of a dual solution $Y$ and a concise classical description $(Y',z)$ of a primal solution $M^{(i)} \propto e^{Y' - z\rho^i}$, such that the $M^{(i)}$s form a close to optimal POVM. Note that this representation gives an interesting way of compressing a POVM, since the $Y'$ matrix is only $d\times d$, and encodes $k$ POVM operators with the help of the $\rho^{(i)}$ matrices. The dual solution $Y$ could be of independent interest too, solving the following problem: for a set of density operators, find the matrix with the smallest trace that is psd bigger then all given density operators. 

\paragraph*{A lower bound.} To find a lower bound, fix $d=2$, i.e., consider a single qubit. Now let $z\in \{0,1\}^k$ be the input for a search problem, we want to distinguish the cases $|z| = 0$ and $|z| = 1$ under the promise that we are in one of these cases. This is known to take $\Omega(\sqrt{k})$ quantum queries or $\Omega(k)$ classical queries. Now let $\rho^{(j)} = \ketbra{z_j}{z_j}$. Given query access to $z$ it is easy to construct the input oracles for any of the three input models. Clearly if $z=0^k$ then all states are equal thus the total success probability is always $1$. However, if $z_k = 1$ then by setting $M^{(k)}:=\ketbra{1}{1}$ and choosing the other measurement operators arbitrarily, we clearly get a total success probability of $2$. Hence a 1/3-approximation to the optimal value of the SDP given above will solve the search problem and hence takes at least $\Omega(\sqrt{k})$ quantum queries or at least $\Omega(k)$ classical queries. 

\subsection{Optimal design}
In the optimal design setting we want to learn a hidden state $\theta \in \mathbb{R}^d$ through experiments. There is a set of $k$ possible experiments, represented by unit vectors $u^{(1)},\dots,u^{(k)} \in \mathbb{R}^d$, and when we execute the $i$th experiment we learn $\langle \theta, u^{(j)}\rangle$ with some noise. In particular we get a sample from $\mathcal{N}(\langle \theta, u^{(j)}\rangle,\sigma_j)$. Precise estimation of $\theta$ requires a lot of experiments, and the problem in optimal design is to decide which distribution to use when choosing the experiments in order to ``minimize'' the covariance matrix of the maximum likelihood estimator of $\theta$. Since the variance of the maximum likelihood estimator is hard to express analytically, we instead look at the \emph{Fischer information matrix}, which is a good approximation for the inverse of the covariance matrix, and has a nice closed form:
\[
F_p = \sum_{i=1}^k p_i u^{(i)}u^{(i)T} / \sigma_i^2,
\] 
where $p_i$ is the probability of doing experiment $u^{(i)}$. Now, to get the covariance matrix ``small'' we would like to get the Fischer information matrix ``large''. For a more detailed explanation, see for example \cite{silvey:optdes}.

The precise meaning of ``small'' and ``large'' can be defined in several sensible ways. The most common criteria are called \emph{A-optimal}, \emph{D-optimal} and \emph{E-optimal design}. In A-optimal design we want to minimize the sum of the eigenvalues of the covariance matrix, or as an approximation the sum of the eigenvalues of the inverse of the Fischer information matrix. Unfortunately the SDP formulation of this problem has parameters $r,R$ that make our methods inefficient. In D-optimal design we want to minimize the determinant of the covariance matrix, this can be approximated with a convex program, but sadly this problem does not naturally correspond to an SDP. 

We will consider E-optimal design. In this setting we would like to minimize the operator norm of the covariance matrix. Since this is hard to do, we will try to maximize the smallest eigenvalue of the Fischer information matrix. Let $P := \frac{1}{ d\max_i \sigma_i^{2}}$ be an input parameter dependent on the precision of the experiments averaged over the coordinates.

\begin{theorem}
	The E-optimal design problem, that is, finding a distribution $p$ such that the smallest eigenvalue of $F_p$ is maximized, can be solved up to additive error $\eps$ using sparse access to the $s$-sparse experiment (unit) vectors $u^{(1)},\dots,u^{(k)}\in \mathbb{R}^d$ and oracle access to the $\sigma_i$ values with
	\[
	\bOt{\left(\sqrt{k}+\sqrt{d}\frac{P^2}{\eps}\right)s\frac{P^8}{\eps^4}}
	\]
	queries on a quantum computer.\footnote{The dependence on $s$ can be reduced to $\sqrt{s}$ by using state preparation and the quantum operator model.}
\end{theorem}

\begin{proof}
	We consider the following SDP:
	\begin{align*}
	\max \ \ \ & t\\
	\text{s.t.} \ \ \ & \sum_{i=1}^k p_i u^{(i)}u^{(i)T} / \sigma_i^2 \succeq t I_d\\
	& \sum_{i=1}^k p_i \leq 1\\
	& p_i \geq 0 \text{ for all }i\in [k]
	\end{align*}
	Clearly this SDP would maximize the minimal eigenvalue of $F_p$.
	We can rewrite this in standard dual form, flipping the sign of the optimal value:
	\begin{align*}
	\min \ \ \ & -t\\
	\text{s.t.} \ \ \ & \sum_{i=1}^k p_i \begin{bmatrix} -1 & \\ & u^{(i)}u^{(i)T} / \sigma_i^2 \end{bmatrix} + t \begin{bmatrix}0 & \\ & -I_d \end{bmatrix} \succeq \begin{bmatrix} -1 & \\ & 0 \end{bmatrix} \\
	& p_i \geq 0 \text{ for all }i\in [k]\\
	& t \geq 0
	\end{align*}
	The corresponding primal problem is then easy to write down:
	\begin{align*}
	\max \ \ \ & -z\\
	\text{s.t.} \ \ \ & - z + \tr{X u^{(i)}u^{(i)T}} / \sigma_i^2 \leq 0 \text{ for all }i\in[k]\\
	&\tr{X} \geq 1\\
	& z \geq 0, X\succeq 0 
	\end{align*}
	From the size of the input it follows that $n=1+d$ and $m=1+k$ for this SDP. Furthermore, the row sparsity of the constraint matrices is equal to the vector sparsity of the $u^{(i)}$, which justifies the use of $s$ for the sparsity of the vectors $u^{(i)}$. It remains to give a bound on $r$ and $R$. Note that the trace constraint on $X$ will be tight for an optimal $X$ and hence $R = 1+|\opt|$, where $\opt$ is the optimal value of one of these SDPs. Similar for the sum constraint on $p$, we get $r =1+ |\opt|$. To give a bound on $|\opt|$ we rewrite the primal again, flip the sign of the optimum and flipping the sign of $z$:
	\begin{align*}
	\min \ \ \ & z\\
	\text{s.t.} \ \ \ &  \tr{X u^{(i)}u^{(i)T}} / \sigma_i^2 \leq z \text{ for all }i\in[k]\\
	&\tr{X} \geq 1\\
	& z \geq 0, X\succeq 0 
	\end{align*}
	Now, let us construct a feasible point, since we have a minimization SDP, this will give an upper bound on $|\opt|$. Let $X = I_d / d$, then $\tr{X u^{(i)}u^{(i)T}} = 1 / d$, so picking $z = \max_i \frac{ 1}{d\sigma_i^2}$ will give a feasible point with objective value $z$. We conclude that $r = R = \bigO{\frac{1}{ \max_i d \sigma_i^{2}}}$ suffices. The stated complexity follows using our complexity bounds on quantum SDP-solving.
\end{proof}

\section{Lower bounds for the new input models}

Previous work~\cite{AGGW:SDP} showed an $\Omega\left(\sqrt{\max\{n,m\}} \min\{n,m\}^{3/2}\right)$ lower bound for the quantum query complexity of SDP-solving in the sparse input model. For this bound $s=1$, $\eps = 1/3$ and $R=r=\min\{n,m\}^2$. By letting either $n$ or $m$ be constant, the $\Omega(\sqrt{n}+\sqrt{m})$ lower bound from \cite{brandao:quantumsdp} can be recovered. The improved upper bounds of this paper show that the dependence on $n$ and $m$ is tight up to logarithmic factors. It remains an open question whether a lower bound with an interesting dependence on $s$ and $Rr/\eps$ can be proven.

In this section we prove lower bounds for the new input models: the quantum state model and the quantum operator model. To do so, we first prove a lower bound in the Hamiltonian input model, where we can time-evolve under the matrices $A_j$, see Definition~\ref{def:HamImp}. In all cases the goal is to show that the term $\sqrt{m}/\eps$ times the relevant normalization parameter (for example $B$ in the quantum state model) is necessary.

\begin{definition}[Hamiltonian input model]\label{def:HamImp} In the \emph{Hamiltonian} input model for SDPs, we have access to two oracles for the $A_j$ matrices. The first oracle, $O_t$, gives a classical description of a real vector $t \in \mathbb{R}^j$ in the usual way 
	\[
	O_t \ket{j}\ket{0} = \ket{j}\ket{t_j}.
	\]
	The second oracle, $O_H$, performs the Hamiltonian simulation with $A_j$ for time $1/t_j$:
	\[
	O_H \ket{j} \ket{\psi} = \ket{j} e^{iA_j/t_j} \ket{\psi}
	\]
	Alongside the oracles we also require an upper bound $\tau \geq \max_j t_j$ as part of the input for an SDP. As in the other input models, we assume that we can also apply the inverse of the oracles.
\end{definition}

Now we invoke a result of Gilyén et al.~\cite[Theorem 2]{gilyen:qgradient}, which gives a lower bound on the number of queries needed for distinguishing different phase oracles. In the spirit of Definition~\ref{def:HamImp}, we will view this as the task of distinguish a set of diagonal Hamiltonians.

\begin{theorem}[Hybrid method for arbitrary phase oracles]\label{thm:arbHybLow}
	Let $G$ be a (finite) set of labels and let $\mathcal{H}:=\mathrm{Span}(\ket{x}\colon x\in G)$ be a Hilbert space. For a function $\tilde{f}:G\rightarrow \R$ let $\mathrm{O}_{\!\tilde{f}}$ be the phase oracle acting on $\mathcal{H}$ such that $$\mathrm{O}_{\!\tilde{f}}:\ket{x}\to e^{i \tilde{f}(x)}\ket{x} \quad \text{ for every } x\in G.$$
	Suppose that $\F$ is a finite set of functions $G\rightarrow \R$, and the function $f_*\colon G\rightarrow \R$ is not in $F$. If a quantum algorithm makes $T$ queries to a (controlled) phase oracle $\mathrm{O}_{\!\tilde{f}}$ (or its inverse) and for all $f\in \F$ can distinguish with probability at least $2/3$ the case $\tilde{f}=f$ from the case $\tilde{f}=f_*$, then
	$$  T \geq \frac{\sqrt{|\F|}}{3}\left/\sqrt{\max_{x\in G}\sum_{f\in \F}\left|f(x)-f_*(x)\right|^2}\right..$$
\end{theorem}

Now the following lower bound follows naturally by reducing the above ``Hamiltonian discrimination problem'' to solving an SDP in the Hamiltonian input model.
\begin{lemma}
Let $\eps\in(0,1/2]$, $2\leq m$ and $1\leq \tau$. Then there is an LP (and hence an SDP) (with $R,r=  \bigO{1}$) for which an $\eps$-approximation of the optimal value requires $\Omega(\sqrt{m} \frac{\tau }{\eps} )$ queries to $O_H$ in the Hamiltonian input model.
\end{lemma}
\begin{proof}
Let $H_1,\dots,H_m \in \mathbb{R}^{2\times 2}$ be such that 
\begin{enumerate}[label=(\alph*)]
	\item\label{it:a} either all $H_j$ are $I/ (2\tau)$,
	\item\label{it:b} or all but one matrices are $I/ (2\tau)$, and there is one $H_j$ that is equal to
	\[
	\begin{bmatrix}
	1/(2\tau) + \eps/\tau & 0\\
	0&1/(2\tau) -\eps/\tau
	\end{bmatrix} = \frac{1}{2\tau}I + \frac{\eps}{\tau} Z.
	\]
\end{enumerate}	
Let us assume that we have access to the phase oracle $\mathrm{O}:\ket{j}\ket{b}\to e^{i (H_j)_{bb}}\ket{j}\ket{b}$. In case  \ref{it:b} there are $m$ possible different choices for this oracle. It is easy to see by Theorem~\ref{thm:arbHybLow} that distinguishing case \ref{it:a} form \ref{it:b} requires $\Omega\left(\sqrt{m}\frac{\tau}{\eps}\right)$ queries to $\mathrm{O}$.

Now we show that using the above phase oracles we can define an SDP in the Hamiltonian input model, solving which to $\eps$-precision distinguishes case \ref{it:a} form \ref{it:b}, proving the sought lower bound.

Let us define $A_j:=\tau H_j$ (and $t = (\tau,\dots,\tau)$), so all the $A_j$'s are either $I/2$ or $I/2 +\eps Z$, furthermore let
\[
C = \begin{bmatrix}1&0\\0&0\end{bmatrix},
\]
and $b = (1,\dots,1)^T$ the all-one vector.
Note that since at least one of the $A_j$ matrices is $I/2$, we know that $R=2$ suffices as an upper bound on the trace. Furthermore, it is easy to verify from the dual that $r = 2$ suffices as well. 

Now we analyze the optimal value. If all $A_j$ matrices are $I/2$ then all constraints are the same:
\[
   X_{11}/2 + X_{22}/2 \leq 1 
\]
from which it clearly follows that $\opt = 2$.

If one $A_j$ matrix is not $I/2$, then the constraint
\[
  (1/2+\eps) X_{11} + (1/2-\eps)X_{22} \leq 1 
\]
is present.  It follows that
\[
  (1/2+\eps)X_{11} \leq 1 \Rightarrow X_{11} \leq \frac{1}{1/2+\eps} 
\]
which will clearly be tight in the optimum. Using that 
\[
2-4\eps \leq \frac{1}{1/2+\eps}  \leq 2-2\eps
\]
we conclude that $2-4\eps \leq \opt \leq 2-2\eps$.

Hence solving this SDP up to precision $\eps$ will distinguishing case \ref{it:a} form \ref{it:b} and hence requires $\Omega\left(\sqrt{m}\frac{\tau}{\eps}\right)$ queries.
\end{proof}

To prove the lower bounds for the quantum state model and the quantum operator model we reduce the Hamiltonian input model to them.
\begin{lemma} 
Let $\eps\in(0,1]$. Given an SDP in the Hamiltonian input model with parameter $\tau\geq 2$ 
(for technical reasons also assume that $t_j\geq 2$ for all $j$),
 an $\eps$-approximate oracle call in the quantum operator model with $\alpha=2\tau$ can be simulated using $\BOT{\eps}{1}$ queries. 
\end{lemma}
\begin{proof}
For simplicity let us drop the index $j$.
Let $H':=H/t$, and apply Theorem~\ref{thm:Taylor} to $H'$ with $f(x)=x$, setting $x_0=0$, $r=1$, $\delta=\pi/2-1$, and $K=2$, providing a $(2,\BOT{\eps}{1},\eps)$-block encoding of $H'$ via controlled $\left(\bigO{\log\left(1/\eps\right)},1\right)$-Hamiltonian simulation, which can be easily implemented by $\bigO{\log\left(1/\eps\right)}$ controlled oracle calls.
\end{proof}
\begin{corollary}
Let $\eps\in(0,1/2]$, $2\leq m$ and $2\leq \alpha$. Then there is an LP (and hence an SDP) (with $R,r=  \bigO{1}$) for which an $\eps$-approximation of the optimal value requires $\Omega(\sqrt{m} \frac{\alpha }{\eps} )$ queries to $O_H$ in the quantum operator model.
\end{corollary}

For the quantum state input model we only give a reduction for LPs, i.e., the case where all input matrices are diagonal.
\begin{lemma} 
Let $\eps\in(0,1]$. Given an LP in the Hamiltonian input model with parameter $\tau\geq4$
(for technical reasons also assume that $t_j\geq 4$ for all $j$),
then an $\eps$-approximate oracle call in the quantum state model with $B=n\tau$ can be simulated using $\BOT{\eps}{1}$ queries.
\end{lemma}
\begin{proof} 
For simplicity let us drop the index $j$.
Let $H':=I\pm H/t$, and apply Theorem~\ref{thm:Taylor} to $H'$ with $f(x)=\sqrt{x/2}$, setting $x_0=1$, $r=1/2$, $\delta=\pi/6-1/2$, and $K=1$, providing a $(1,\BOT{\eps}{1},\bigO{\eps})$-block encoding of $\sqrt{(I+ H')/2}$, 
via controlled $\left(\bigO{\log\left(1/\eps\right)},3\right)$-Hamiltonian simulation, which can be easily implemented by $\bigO{\log\left(1/\eps\right)}$ controlled oracle calls.
The validity of the $K=1$ bound follows from the observation that 
$$\sqrt{(1+x)/2}
=\frac{1}{\sqrt{2}}\sqrt{1+x}
=\frac{1}{\sqrt{2}}\sum_{k=0}^{\infty}\binom{1/2}{k}x^k
,$$
thus
\begin{align*}
\frac{1}{\sqrt{2}}\sum_{k=0}^{\infty}\left|\binom{1/2}{k}\right|(r+\delta)^k
&=\frac{1}{\sqrt{2}}\sum_{k=0}^{\infty}\left|\binom{1/2}{k}\right|(\pi/6)^k
=\sqrt{(1+\pi/6)/2}\leq 1=:K.
\end{align*}

The state input oracles can be implemented as follows: controlled on the state $\ket{\pm}$ we apply $\sqrt{(I\pm H/t)/2}$ to the first half of the state $\sum_{i=1}^{n}\ket{i}\ket{i}/\sqrt{n}$, 
resulting in subnormalized density operators $\varrho_{\pm}=(I\pm H/t)/(2n)$, so that $\varrho_{+}-\varrho_{-}=H/(nt)$.
If needed one can further subnormalize the $\rho$-s in order to get a uniform normalization factor $n\tau$ (instead of $nt_j$).
\end{proof}
\begin{corollary}
Let $\eps\in(0,1/2]$, $2\leq m$ and $1\leq B$. Then there is an LP (and hence an SDP) (with $R,r=  \bigO{1}$) for which an $\eps$-approximation of the optimal value requires $\Omega(\sqrt{m} \frac{B }{\eps} )$ queries to $O_H$ in the quantum state model.
\end{corollary}

\paragraph{Acknowledgments.} 
We thank the authors of~\cite{brandao:expSDP2} for sending work-in-progress versions of their paper, and Fernando Br\~{a}ndao, Tongyang Li and Xiaodi Wu for personal communication. A.G. thanks Robin Kothari and Nathan Wiebe for useful discussions. We thank Jamie Sikora for useful discussions about applications and for suggesting the state discrimination problem. We are grateful to Ronald de Wolf and Sander Gribling for useful discussions, and advice about the manuscript.

\bibliographystyle{alphaUrlePrint}
\bibliography{qc.bib}

\appendix

\section{Implementing smooth functions of Block-Hamiltonians}\label{apx:HamSim}

In~\cite[Appendix B]{AGGW:SDP} techniques were developed that make it possible to implement smooth-functions of a Hamiltonian $H$, based on Fourier series decompositions and using the Linear Combinations of Unitaries (LCU) Lemma~\cite{BerryChilds:hamsimFOCS}. The techniques developed in~\cite[Appendix B]{AGGW:SDP} access $H$ only through controlled-Hamiltonian simulation, which is defined in the following way:

\begin{definition}\label{def:controlledSim}
	Let $M=2^J$ for some $J\in \mathbb{N}$, $\tau\in\mathbb{R}$ and $\epsilon\geq0$. We say that the unitary 
	$$
	W:=\sum_{m=-M}^{M-1}\ketbra{m}{m}\otimes e^{im\tau H}
	$$ 
	implements controlled $(M,\tau)$-simulation of the Hamiltonian $H$, where $\ket{m}$ denotes a (signed) bitstring $\ket{b_Jb_{J-1}\ldots b_0}$ such that $m=-b_J2^J+\sum_{j=0}^{J-1}b_j2^j$. 
\end{definition}

The main theorem of \cite[Appendix~B]{AGGW:SDP} states the following:

\begin{theorem}[{Implementing a smooth function of a Hamiltonian~\cite[Theorem~40]{AGGW:SDP}}]\label{thm:Taylor}
	Let $x_0\in\mathbb{R}$ and $r>0$ be such that $f(x_0+x)=\sum_{\ell=0}^{\infty} a_\ell x^\ell$ for all $x\in\![-r,r]$. 
	Suppose $K>0$ and $\delta\in(0,r]$ are such that $\sum_{\ell=0}^{\infty}(r+\delta)^\ell|a_\ell|\leq K$. 
	If $\nrm{H-x_0I}\leq r$ and $\eps'\in\!\left(0,\frac{1}{2}\right]$, then we can implement a unitary $\tilde{U}$ that is a $(K,a+\bigO{\log(r\log(1/\eps')/\delta)},K \eps')$-block-encoding of $f(H)$, with a single use of a circuit $V$ which is a $(1,a,\eps'/2)$-block-encoding of controlled $\left(\bigO{r\log(1/\eps')/\delta},\frac{\pi}{2(r+\delta)}\right)$-simulation of $H$, and using $\bigO{r/\delta\log\left(r/(\delta\eps')\right)\log\left(1/\eps'\right)}$ two-qubit gates. 
\end{theorem}

This theorem can in particular be applied to the function $f(x) = x^{-c}$, a power function of a Hamiltonian with negative exponent:

\begin{corollary}[\cite{CGJ:PowerOfBlockPowers18}]\label{cor:NegativePowerCost}
	Let $\kappa\geq 2$, $c\in(0,\infty)$ and $H$ be a $w$-qubit Hamiltonian such that $I/\kappa\preceq H \preceq I$.\\
	Then we can implement a unitary $\tilde{U}$ that is a $(2\kappa^{c},a+\bigO{\log(\kappa^c\max\left(1,c\right)\log(\kappa^c/\eps))}, \eps)$-block-encoding of $H^{-c}$,  with a single use of a circuit $V$ which is a $(1,a,\eps/(4\kappa^c))$-block-encoding of controlled $\left(\bigO{\kappa\max\left(1,c\right)\log(\kappa^c/\eps)},\bigO{1}\right)$-simulation of $H$, and using $\bigO{\kappa\max\left(1,c\right)\log^2\left(\kappa^{1+c}\max\left(1,c\right)/\eps\right)}$ other two-qubit gates. 
\end{corollary}
We do not proof this here, but a full proof can be found in~\cite{CGJ:PowerOfBlockPowers18}.
Similarly they get a result about implementing power functions of positive exponents.
\begin{corollary}[\cite{CGJ:PowerOfBlockPowers18}]\label{cor:PositivePowerCost}
	Let $\kappa\geq 2$, $c\in(0,1]$ and $H$ be an $s$-qubit Hamiltonian such that $I/\kappa\preceq H \preceq I$.\\
	Then we can implement a unitary $\tilde{U}$ that is a $(2,a+\bigO{\log\log(1/\eps)}, \eps)$-block-encoding of $H^{c}$, with a single use of a circuit $V$ which is a $(1,a,\eps/4)$-block-encoding of controlled $\left(\bigO{\kappa\log(1/\eps)},\bigO{1}\right)$-simulation of $H$, and using $\bigO{\kappa\log\left(\kappa/\eps\right)\log\left(1/\eps\right)}$ other two-qubit gates. 
\end{corollary}

We can use these results to show that for a quantum state $\rho$ it is possible to sample from a random variable with expectation value $\theta$-close to $\tr{A\rho}$ and with standard deviation $\sigma\leq6$, using $(\BOT{\theta}{1},\bigO{1})$-controlled Hamiltonian simulation of $A$. For more details see Definition~\ref{def:traceEstimator}.
\begin{corollary}\label{cor:traceCalc}
  Suppose $-I\preceq A \preceq I$, $0 < \theta < 1$. We can implement a trace estimator for $A$ with standard deviation $\sigma \leq 6$ and bias $\leq\theta$ with a single use of an $(\BOT{\theta}{1},\bigO{1})$-controlled Hamiltonian simulation circuit for $A$, and with $\BOT{\theta}{1}$ other two-qubit gates. 
\end{corollary}
\begin{proof}
	We apply Corollary~\ref{cor:PositivePowerCost} to the square root function with the operator $I/2+A/4$ (so that $\kappa\geq 4$), giving a $(2,\bOt{1},\Theta(\theta))$-block-encoding $\tilde{U}$ of $\sqrt{I/2+A/4}$ using an $(\BOT{\theta}{1},\bigO{1})$-controlled Hamiltonian simulation circuit for $I/2+A/4$. The probability of finding the ancilla qubit to be $\ket{0}$ upon measurement is 
	\begin{align*}
		\tr{\left(I\otimes\bra{0}\right)\tilde{U}^\dagger \left(\rho\otimes\ketbra{0}{0}\right)\tilde{U}\left(I\otimes\ket{0}\right)}
		&=\mathrm{Tr}\Big(
			\underset{\approx\frac{\sqrt{I/2+A/4}}{2}}{\underbrace{\left(I\otimes\bra{0}\right)\tilde{U}\left(I\otimes\ket{0}\right)}}
			\underset{\approx\frac{\sqrt{I/2+A/4}}{2}}{\underbrace{\left(I\otimes\bra{0}\right)\tilde{U}^\dagger\left(I\otimes\ket{0}\right)}}
			\rho\Big)\\
		&=
		\frac{1}{8} + \frac{\tr{A\rho}}{16}+ \bigO{\theta}.
	\end{align*}
	Upon measuring the ancilla qubit and getting outcome $\ket{0}$ we output $16-2=14$. In case of any other measurement outcome we output $-2$. By choosing the right constants so that $\tilde{U}$ is a precise enough block-encoding we can ensure that the bias is less than $\theta/2$, and the standard deviation $\sigma\leq 6$.
	
	Finally observe that a controlled Hamiltonian simulation circuit for $I/2+A/4$ can be implemented as a product of controlled Hamiltonian simulation circuits for $I/2$ and $A/4$. 
\end{proof}

The following lemma from~\cite{CGJ:PowerOfBlockPowers18} gives a way to implement the controlled Hamiltonian simulation used in this appendix, provided a block-encoding of $H$, as used in this paper.
In particular, it shows how to make Theorem~\ref{thm:blockHamSim} controlled by a time parameter.

\begin{lemma}[\cite{CGJ:PowerOfBlockPowers18}]\label{lemma:controlledHamsin}
	Let $M=2^J$ for some $J\in \mathbb{N}$, $\tau\in\mathbb{R}$ and $\epsilon\geq0$. Suppose that $U$ is an $(\alpha,a,\eps/|2 (J+1)^2 M \tau|)$-block-encoding of the Hamiltonian $H$. Then we can implement a $(1,a+2,\eps)$-block-encoding of a controlled $(M,\tau)$-simulation of the Hamiltonian $H$, with $\bigO{|\alpha  M\tau|+J\frac{\log(J/\eps)}{\log\log(J/\eps)}}$ uses of controlled-$U$ or its inverse and with $\bigO{a|\alpha M\tau|+aJ\frac{\log(J/\eps)}{\log\log(J/\eps)}}$ three-qubit gates.	
\end{lemma}

\end{document}